\documentclass[reqno, 10pt]{amsart}%
\usepackage{amsfonts}
\usepackage{amsmath}
\usepackage{amssymb}
\usepackage{enumerate}
\usepackage{xspace}
\usepackage{graphicx}
\usepackage{color}%
\numberwithin{equation}{section}
\providecommand{\U}[1]{\protect\rule{.1in}{.1in}}
\newtheorem{thm}{Theorem}[section]
\newtheorem{lem}[thm]{Lemma}
\newtheorem{cor}[thm]{Corollary}

\newtheorem{defin}[thm]{Definition}

\newtheorem{rem}[thm]{Remark}

\newcommand\R{{\mathbb R}}

\newcommand\1{{1\kern-.25em\hbox{\rm I}}}
\newcommand\eu{{1\kern-.25em\hbox{\sm I}}}

\newcommand\MM{{\mathcal M}}

\newcommand\e{\epsilon}

\newcommand{\nada}[1]{}

\begin{document}
\title[dynamics and kinetic limit for Vicsek-type particles]{Dynamics and kinetic limit for a system of noiseless $d$-dimensional
Vicsek-type particles}
\author{Michele Gianfelice}
\address{Michele Gianfelice, Dipartimento di Matematica, Universit\`{a} della Calabria,
Campus di Arcavacata, Ponte P. Bucci - cubo 30B\\
87036 Arcavacata di Rende (CS), Italy.}
\email{\texttt{gianfelice@mat.unical.it}}
\author{Enza Orlandi}
\address{Enza Orlandi, Dipartimento di Matematica\\
Universit\`a di Roma Tre\\
L.go S.Murialdo 1, 00146 Roma, Italy. }
\email{\texttt{orlandi@mat.uniroma3.it}}
\thanks{M.G. supported by MIUR - PRIN 2011-2013: Campi aleatori, percolazione ed
evoluzione stocastica di sistemi con molte componenti}
\thanks{E.O. supported by MIUR - PRIN 2011-2013: Studio di sistemi disordinati con
metodi probabilistici e loro simulazione numerica, and ROMA TRE University}
\date{\today }

\begin{abstract}
We analyze the continuous time evolution of a $d$-dimensional system of $N$
self propelled particles with a kinematic constraint on the velocities
inspired by the original Vicsek's one \cite{VCB-JCS}. Interactions among
particles are specified by a pairwise potential in such a way that the
velocity of any given particle is updated to the weighted average velocity of
all those particles interacting with it. The weights are given in terms of the
interaction rate function. The interaction is not of mean field type and the
system is non-Hamiltonian. When the size of the system is fixed, we show the
existence of an invariant manifold in the phase space and prove its
exponential asymptotic stability. In the kinetic limit we show that the
particle density satisfies a nonlinear  kinetic equation of Vlasov type, under suitable conditions on
the interaction. We study the qualitative behaviour of the solution and we
show that the Boltzmann-Vlasov entropy is strictly decreasing in time.

\end{abstract}
\maketitle

\section{Introduction}

The analysis of a network of a large number of coordinated self propelled
particles (agents) is a sub discipline of control theory which has seen a
rapid development during the last decade \cite{BDT, W, JLM, CKFL, B-NVR, CS,
CHDB}. This is due to its several potential application in understanding the
collective behavior in biological systems (for example fish schools and bird
flocks) \cite{CKFL}, computer science \cite{R, BDT}, engineering \cite{JLM, CS, CHDB},
economy \cite{DY} and social sciences \cite{W, B-NVR}. Explaining the
emergence of these coordinated movements in terms of microscopic decisions of
each individual member of a network is a hot matter of research in the natural
sciences. To model the particle self-organized behavior one assigns to each
particle a simple communication/interaction rule in order for the whole system
to dynamically reproduce, in a given regime of the model's parameters,
specific phase space patterns.

The emergence of phase space patterns persistent in time described by a large
connected cluster of coherently moving particles is called flocking or
swarming (also schooling or herd behavior). Basic models of flocking behavior
generally follow three simple rules: 1) separation, that is to avoid crowding
neighbors (usually modeled by short range repulsive interactions); 2)
alignment, i.e. to steer towards average heading of neighbors; 3) cohesion,
i.e. to steer towards the average position of neighbors (usually modeled by
long range attractive interactions).

The seminal work in the direction of modeling flocking behavior is the one of
Vicsek et al. \cite{VCB-JCS}. They proposed a model of $N$ interacting
particles located on a $2$-dimensional torus of diameter $D.$ The velocity of
each given particle belongs to the unit circle and at each time step its
direction is updated at the empirical average of the velocity's directions of
all the particles lying in a neighborhood of radius $1$ from the given one,
including itself, plus a random perturbation. Particles positions are then
updated according to their velocity. Computer simulations proved that, when
the particle density $\frac{N}{D^{2}}$ is sufficiently high and the noise
intensity sufficiently small, the distribution of the velocities of the
particles concentrates around the velocity of the barycenter of the system,
although this is not a quantity preserved by the dynamics.

We propose a simple model of continuous time noiseless multi-agent evolution
closely inspired to the original Vicsek's one. The particles interact
(communicate) with each other trough a pairwise interaction potential in such
a way that the velocity of any given particle is updated to the weighted
average velocity of all those particles communicating with it. This choice
makes the interaction not of mean field type. Furthermore, the system is non
Hamiltonian. As a result, there is a tendency of neighboring particles to
align their velocities. This is the crucial element in the mechanism of the
emergency of a coherent motion.

For what concerns flocking behaviour our model takes into account alignment
and cohesion, but violates the separation rule since the particles can overlap.

We prove for such model two type of results. First, we analyze the $N$
particle dynamics in $\mathbb{R}^{d}.$ We show that there exists an invariant
manifold in the phase space and prove exponential asymptotic stability of the
invariant manifold when the initial conditions for particles dynamics are
suitably chosen. This implies that the system, under the chosen initial
conditions, will reach a state of flocking. Then, we study the kinetic limit
($N\rightarrow\infty$) of the system. Since the interaction is not of mean
field type care needs to be taken in the definition of the velocity field in
the phase space and consequently in the evolution of the particle density. We
explain in more details how to deal with these difficulties in Section 4. We
prove that the particle density satisfies a nonlinear equation of Vlasov type 
when the particles are confined on a torus and subject to a \emph{short-range
potential}\ of Gaussian type. Similar result holds in $\mathbb{R}^{d}$ when
the interaction among particles is given by a suitable regularization of a
finite range potential. We further show that the Boltzmann-Vlasov entropy is
strictly decreasing in time. As a consequence, one can argue that, even if the initial
distribution of the particles is absolutely continuous w.r.t. Lebesgue
measure, the limit density distribution is singular w.r.t. Lebesgue measure.
This is consistent with what one expects from the model. For time long enough
the position and velocity particle distribution will concentrate on specific
phase-space patterns.

A continuous time version of Vicsek's model, as well as its stochastic
counterpart driven by the Brownian motion, has been proposed in \cite{DM} and
the corresponding kinetic equations heuristically derived and studied. In
fact, at present time, to our knowledge, a rigorous derivation and analysis of
Vicsek's model kinetics, as well as hydrodynamics, is lacking.

Another basic model for flocking is the Cucker-Smale one \cite{CS}. In this
and related models \cite{DCBC, AIR} the variation in time of the momentum of a
given particle is the weighted sum of the differences between the particle's
momentum and those of the other system's components, with weights depending of
the relative distances among particles divided by the total number of
particles $N.$ It is worth notice that, for all these models, the interaction
between two given particles is of order $1/N,$ therefore when the size of the
system becomes large, particles tend to decorrelate. On the contrary, in the
original Vicsek's model, the interaction between a given couple of particles
is of order one. Moreover, Cucker-Smale dynamics preserves the velocity of the
barycenter, which is not the case for Vicsek's. The order of the interaction
with respect to the size of the system is the peculiar feature distinguishing
Vicsek's from Cucker-Smale algorithm.\ Therefore, in our opinion, variants of
the Cucker-Smale momenta updating rule taking into account only the
differences among the directions of the momenta of the particles, rather than
those of the momenta as vectors, are somewhat improperly ascribed to variants
of the Vicsek's model \cite{BCC2}. Cucker-Smale and related models have been
more deeply investigated in the mathematical literature and their mean-field
limit equations rigorously derived and studied in \cite{HT, HL, CFRT, CCR,
AIR} in the noiseless case and in \cite{BCC1, BCC2} in the stochastic case
driven by Brownian motion. Moreover, the hydrodynamics equations for these
models have also been rigorously studied but formally derived \cite{HT, CDP,
CCR}.

Recently, a model analogous to the one we propose in this paper, but with
communication rate function restricted to the Cucker-Smale model one has been
introduced and analysed in \cite{MT}. The authors prove that the strategy
originally proposed to study the emergence of flocking behaviour for a system
of self-propelled particle updating their velocity with the standard
Cucker-Smale algorithm also applies to this case with the same restrictions on
the decay of the communication rate function. It turns out that the model we
propose in the present paper is more general than the one proposed in
\cite{MT}, since includes also sufficiently smooth compactly supported
communication rate functions, and so are the results about the emergence of
flocking behaviour for the particle system.

The plan of the paper is the following. In Section 2 we describe the model,
set the notations and present the main results. In Section 3 we analyze the
system when the number of particles is fixed. In Section 4 we analyze the
system when the number of particles goes to infinity. In the appendix we
collect proofs of results used along the previous sections.

\vskip0.5cm

\noindent\textbf{Acknowledgements}. Enza Orlandi thanks Carlangelo Liverani
for useful discussions. Michele Gianfelice thanks Fabio Fagnani and Marco
Isopi for interesting discussions on the subject and Seung-Yeal Ha for useful
discussions and comments as well as for pointing out reference \cite{MT}.

\section{Description of the model, notation and results}

\subsection{Notations}

Given $x\in\mathbb{R}^{d},d\geq1,$ we denote by $x^{i}$ its $i$-th component,
$i=1,..,d,$ with respect to the canonical basis $\left(  e_{1},..,e_{d}%
\right)  .$ For any $x,y\in\mathbb{R}^{d}$ we set $x\cdot y:=\sum_{i=1}%
^{d}x^{i}y^{i}$ to be the scalar product between $x$ and $y.$ Hence, we denote
by $\left\vert x\right\vert :=\sqrt{x\cdot x}$ the associated Euclidean norm
and by $B_{r}\left(  x\right)  :=\{y\in\mathbb{R}^{d}:\left\vert
y-x\right\vert \leq r\}$ the ball of radius $r>0$ centered at $x$ and
$B_{r}:=B_{r}\left(  0\right)  .$ Furthermore we set $\left\Vert x\right\Vert
_{\infty}:=\max_{i=1,..,d}\left\vert x^{i}\right\vert .$

Given an integer $N\geq2,$  we denote   a point in  $\mathbb{R}^{Nd}$  by  $\mathbf{x:=}\left(  x_{1},..,x_{N}\right)
\in\mathbb{R}^{Nd}$,   its    norm    by $\left\vert \mathbf{x}\right\vert :=\sqrt
{\mathbf{x\cdot y}}$, where   $\mathbf{x}\cdot\mathbf{y}:=\sum_{i=1}%
^{N}x_{i}\cdot y_{i}$.     We  denote by $\mathbf{B}_{r}\left(  \mathbf{x}\right)
:=\{\mathbf{y}\in\mathbb{R}^{Nd}:\left\vert \mathbf{y-x}\right\vert \leq r\}$
the ball of radius $r>0$ centered at $\mathbf{x}.$

Partial derivative w.r.t. any component $x^{i}$ of $x\in\mathbb{R}^{d}$ will
be denoted by $\partial_{x^{i}},$ so that $\nabla_{x}$ stands for $\left(
\partial_{x^{1}},..,\partial_{x^{d}}\right)  $ while, for any $\mathbf{q}%
\in\mathbb{R}^{Nd},$ we set $\nabla_{\mathbf{q}}:=\left(  \nabla_{q_{1}%
},..,\nabla_{q_{N}}\right)  .$

Moreover, we denote by $\mathcal{L}_{n}\left(  \mathbb{R}\right)  $ the  
space of linear operators from $\mathbb{R}^{n}$ to itself and by $\left\Vert
\cdot\right\Vert $ and $\left\Vert \cdot\right\Vert _{\infty}$ the operator
norm induced by respectively the Euclidean and the supremum norm. In
particular $\mathbb{I}_{n},\mathbf{0}_{n}\in\mathcal{L}_{n}\left(
\mathbb{R}\right)  $ denote respectively the identity and the null operator.

\subsection{The model}

Let $N\geq2$ be an integer. We consider $N$ particles of unitary mass in
$\mathbb{R}^{d}$ evolving according to the equations:
\begin{equation}
\left\{
\begin{array}
[c]{l}%
\frac{dq_{i}(t)}{dt}=p_{i}(t)\\
\frac{dp_{i}(t)}{dt}=\frac{\sum_{j=1}^{N}U(q_{i}(t)-q_{j}(t))\left(
p_{j}(t)-p_{i}\left(  t\right)  \right)  }{\sum_{j=1}^{N}U(q_{i}(t)-q_{j}%
(t))}\ ,\;i=1,..,N\\
q_{i}(0)=q_{i}^{0}\ ;\ p_{i}(0)=p_{i}^{0}%
\end{array} 
\right.  \label{mo1}%
\end{equation}
where, for $i=1,..,N,\ (q_{i},p_{i})\in{\mathbb{R}}^{d}\times{\mathbb{R}}%
^{d},\ (q_{i}^{0},p_{i}^{0})$ are the initial conditions and $U$ is a pairwise
interaction. We assume that $U(\cdot)$ is a spherically symmetric positive
function, sufficiently smooth, with support the ball of radius $R$ centred at
zero and so that $U(0)>0$.   This implies that
the denominator in the second equation of (\ref{mo1}) is always strictly
positive.   The choice of $R$ does not play any particular role in the analysis.
Without loss of generality we assume that $\int U(x)dx=1$ and  $\sup_{x\in\mathbb{R}^{d}}U(x)=U(0)$. 
Namely,  for agent-based  models it is reasonable to assume  that self-interaction
is stronger than the interactions between two different particles. 
 A simple example to have in mind for the potential $U$  is $U(x)=  C(d)(1- |x| )\1_{B_1}(x)$,
 where $C(d)$ is  taken such that $\int U(x) dx =1$ or  smoother  versions of this.
In this example,  the particle $q_i$ interacts only with particles    at distance 1.
The vector
field in (\ref{mo1}) is Lipschitz, therefore the existence and the uniqueness
of the solution is granted at least for short time. Since the vector field
increases at most linearly in $w=(\mathbf{q},\mathbf{p})$ the solution
$w(t,w^{0})$ with initial datum $w^{0}$ exists and it is unique for all
$t\geq0.$

To derive the kinetic limit results,  the interaction   must satisfy further requirements
which will be presented and discussed in the following.

\subsubsection{Flocking}

Given a particle configuration $\mathbf{q}\in\mathbb{R}^{Nd}$ we introduce the
notion of communication graph. We use only basic definition of graph theory
useful to define the \textit{flocking behaviour} for the system \eqref {mo1}.
We refer the reader to basic textbooks such as \cite{B} for an account on this subject.

\begin{defin}
Given a particle configuration $\mathbf{q}\in\mathbb{R}^{Nd},$ we define the
\emph{ communication graph} \linebreak%
$\mathcal{G}\left(  \mathbf{q}\right)  :=\left(  \mathcal{V}\left(  \mathbf{q}\right)
,\mathcal{E}\left(  \mathbf{q}\right)  \right)  ,$ where the set of
vertices $\mathcal{V}\left(  \mathbf{q}\right)  =\{q_{1},\dots,q_{N}\}$ is the
collection of the $N$ points of $\mathbb{R}^{d}$ associated to $\mathbf{q}$
and
\begin{equation}
\mathcal{E}\left(  \mathbf{q}\right)  :=\left\{  \left(  q,q^{\prime}\right)
\in\mathcal{V}\left(  \mathbf{q}\right)  \times\mathcal{V}\left(
\mathbf{q}\right)  :U\left(  q-q^{\prime}\right)  >0\right\}
\end{equation}
is the set of edges.
\end{defin}
Two vertices $q$ and $q^{\prime}$ are said to be connected if there are
$q_{1},\dots,q_{k}$  vertices in $\mathcal{V}\left(  \mathbf{q}\right)
,k\in\{{2,\dots,N\},}$ such that $q_{1}=q,q_{k}=q^{\prime}$and $U\left(
q_{i}-q_{i+1}\right)  >0$, for $i=1,\dots,k-1.$
The graph $\mathcal{G}\left(  \mathbf{q}\right)  $ is said to be connected if any two of
its vertices are connected\footnote{Since   $U$  is
spherically symmetric, the communication graph is undirected.  Hence, in this case, the usual
notions of strongly connected graph and connected graph coincide.}.
We will set $\mathcal{V}_{t}:=\mathcal{V}\left(  \mathbf{q}\left(  t\right)
\right)  $ and $\mathcal{G}\left(  t\right)  :=\mathcal{G}\left(
\mathbf{q}\left(  t\right)  \right)  .$

\begin{defin}
\label{flock}The system \eqref {mo1} with initial conditions $w^{0}%
=(\mathbf{q}^{0},\mathbf{p}^{0})$ is said to exhibit a\emph{\ flocking
behavior} if there exists $v\in\mathbb{R}^{d}$ such that, for any
$\epsilon>0,\exists T_{\epsilon}>0:\forall t>T_{\epsilon},$

\begin{itemize}
\item $p_{i}\left(  t, w^{0}\right)  \in B_{\epsilon}\left(  v\right)
,\forall i=1,..,N;$

\item the communication graph $\mathcal{G}\left(  t\right)  $ is connected.
\end{itemize}
\end{defin}

We remark that our definition of emergence of flocking behaviour differs from
the one given for models with long interaction
(e.g.
Cucker-Smale model \cite{HL}, \cite{CFRT}).
 In the  latter  case the
communication graph is always connected, while  this is  not  true for
short range interactions. Therefore we  have  in Definition \ref{flock}   two  conditions, one
on the  particle velocities and  the other on 
the  particle positions.  

\subsection{Results for finite size system}

Let ${\mathcal{I}}$ be the $(N+1)d$ linear manifold
\begin{equation}
{\mathcal{I}}=\cup_{\{v\in{\mathbb{R}}^{d}\}}{\mathcal{I}}(v)\ , \label{v8}%
\end{equation}
where
\begin{equation}
{\mathcal{I}}(v)=\{(\mathbf{q},\mathbf{p})\in{\mathbb{R}}^{dN}\times
{\mathbb{R}}^{dN}:p_{i}=v,i=1,..,N\}\ . \label{d.7}%
\end{equation}
It is immediate to see that ${\mathcal{I}}$ is invariant for the evolution
\eqref  {mo1}. Namely, if the initial data belong to $\mathcal{I}\left(
v\right)  $ the particles evolve independently one from the other with
constant velocity $v.$ The only critical point of the system (\ref{mo1}) is
$\left(  \mathbf{0},\mathbf{0}\right)  .$ We denote, for any $w\in{\mathbb{R}%
}^{2dN},$%
\begin{equation}
\operatorname{dist}\left(  w,{\mathcal{I}}\right)  =\inf_{w^{0}\in
{\mathcal{I}}}|w-w^{0}|
\end{equation}
and by $w(t,w^{1})$   the solution at time $t$ of \eqref {mo1} starting from
$w^{1}\in{\mathbb{R}}^{2dN}.$ We have the following results.

\begin{thm}
\label{sta1a0} The manifold ${\mathcal{I}}$ is stable for the evolution \eqref {mo1}.
\end{thm}

This means that, for any $\epsilon>0,$ there exists $\delta(\epsilon)\le\epsilon$ such
that, for all initial data $w^{0}\in{\mathbb{R}}^{2dN}$  satisfying
$\operatorname{dist}(w^{0},{\mathcal{I}})\leq\delta(\epsilon),$  then 
$\operatorname{dist}(w(t,w^{0}),{\mathcal{I}})\leq\epsilon$ for all $t\geq0$.
Stability of  the manifold ${\mathcal{I}}$ does not imply that the system exhibits a flocking behaviour when starting from $w^0$.
Theorem \ref {sta1a0}   is quite easy to show, see for the proof Corollary
\ref{sta0}.

Next, we show a stronger result.  Assume  that at initial time   the particle   positions are chosen    so   that the communication graph is connected and  their  velocities are conveniently taken;   then, at later times, the particles  will not split into
non interacting groups and the velocity of each one converges
exponentially fast to a velocity vector which is the same for all the $N$
particles. In other words, the system   exhibits  a flocking behaviour,   see
 Definition \ref{flock}.  
\begin{thm}
\label{sta1a} Let $w^{0}=(\mathbf{q}^{0},\mathbf{p}^{0})\in{\mathcal{I}}$ and
assume that the  communication graph $\mathcal{G}\left(
\mathbf{q^{0}}\right)  $ is connected. There exist three positive constants
$r_{0}=r_{0}(w^{0}),T=T\left(  w^{0}\right)  ,\epsilon_{0}=\epsilon_{0}\left(
w^{0}\right)  $ and a set $\mathcal{B}(r_{0},\epsilon_{0},w^{0})\subset
{\mathbb{R}}^{2Nd},$ such that, for any initial datum $w^{1}\in\mathcal{B}%
(r_{0},\epsilon_{0},w^{0})$
\begin{equation}
\operatorname{dist}\left(  w(t,w^{1}),{\mathcal{I}}\right)  \leq\epsilon
_{0}(w^{0})e^{-t\frac{\log2}{T}}\ .
\end{equation}

\end{thm}

The proof of Theorem \ref{sta1a} is presented in Section 3.

\subsection{Results for infinite size system}

The main difficulty in deriving the kinetic limit,   from    system \eqref {mo1} is
that the interaction, see \eqref {mo1},  is not mean field,  i.e. it is not divided by $N,$ the total
particle number. This creates problems in the definition of the evolution of
the particle density since the velocity field in the phase space may be ill
defined if further assumptions on the interaction $U$ and on the configuration
space are not taken into account. We will discuss this point extensively in
Section 4. We overcome these difficulties in two ways. The first way is adding
$\epsilon>0,$ which will be kept fixed, to the denominator of the second
equation of \eqref {mo1}. We keep the interaction $U$ of compact support and
assume for definiteness $\sup_{x}|\nabla U(x)|\leq1.$ We will refer
to the system \eqref {mo1} modified in such a way as $\epsilon-$%
\emph{regularized system}. A second way is to confine the system \eqref{mo1}
in the torus of linear size $D>0,{\mathcal{T}}_{D},$ and taking interactions
$U$ verifying the following assumptions.

\begin{defin}
\label{no1}\textbf{Assumptions on the interaction} Let $\tilde{U}:\mathbb{R}^{d}\rightarrow\mathbb{R}_{+}$ be such that either
\begin{equation} \label {ginevra1}
\sup_{x\in\mathbb{R}^{d}}\left\vert \nabla\log\tilde{U}(x)\right\vert \leq
K
\end{equation}
or
\begin{equation}   
\tilde{U}(x)=\frac{1}{(2\pi R^{2})^{\frac{d}{2}}}e^{-\frac{|x|^{2}}%
{2R^{2}}}\ . \label{inter1}%
\end{equation}
We then define $U$ to be the periodization on the torus ${\mathcal{T}}_{D}$ of
one of the previous $\tilde{U}:$
\begin{equation} \label {ginevra2}
 U(x)=\sum_{n\in{\mathbb{Z}}^{d}}\tilde{U}(x+nD).
\end{equation}

\end{defin}

\begin{rem}
The assumption  \eqref {ginevra1}  is quite strong.  An
interaction $\tilde U$ verifying this assumption should decay for $|x|$
large as $e^{-\frac{|x|}{R}}$,  for some $R>0$. Interactions with compact support do not
satisfy this assumption as well as the interaction \eqref {inter1}.
\end{rem}

 Next, we define the space of measures and the metric we will be using.
    We denote by
${\mathcal{M}}$ 
  the space of probability measures   on $\left(  X\times
B_{1},\mathcal{B}\left(  X\times B_{1}\right)  \right)  $, where the symbol $X$
stands either for ${\mathbb{R}}^{d}$ or for the torus ${\mathcal{T}}_{D},$  $B_{1}$ denotes the ball of radius $1$ in ${\mathbb{R}}^{d}$ and
$\mathcal{B}\left(  X\times B_{1}\right)  $ is the Borel $\sigma$algebra on
$X\times B_{1}.$ We will prove that there is no loss of generality to confine
the velocity in a bounded set and for definiteness we identify this set with
$B_{1}$.  We will be using the same
notations either to denote the space of probability measure on ${\mathcal{T}%
}_{D}\times B_{1}$ or the space of probability measure on ${\mathbb{R}}%
^{d}\times B_{1},$ unless we will have the need to distinguish between the two
configuration spaces in which case we will use the notation ${\mathcal{M}}(X\times
B_{1}).$    In this space we introduce  the bounded Lipschitz distance $d_{b{\mathcal{L}}}$ defined as follows.  
The $d_{b{\mathcal{L}}}$ distance  between two measures $\mu$ and $\nu$ in ${\mathcal{M}}$ is given by
\begin{equation}
d_{b{\mathcal{L}}}(\mu,\nu)=\sup_{g\in{\mathcal{D}}}\left\vert \int
g(x,v)\mu(dx,dv)-\int g(x,v)\nu(dx,dv)\right\vert \ ,
\end{equation}
where
\begin{equation}
{\mathcal{D}}:=\left\{  g\ |\ g:X\times B_{1}\rightarrow\left[  0,1\right]
\ ;\ \left\vert g(x,v)-g(y,w)\right\vert \leq\sqrt{\left\vert v-w\right\vert
^{2}+\left\vert x-y\right\vert ^{2}}\right\}  \ .
\end{equation}
 The  metric $d_{b{\mathcal{L}}}$ generates the weak* topology\footnote{We refer the reader to
\cite{M} for an account on the notion of weak convergence of measures and to
\cite{V} for the relation between the bounded Lipschitz distance and the
Kantorovich-Rubinstein (Wasserstein) distance.} on $\MM$: for a sequence $ \mu^N\in \MM$ and $ \mu \in \MM$
 $$\lim_{N \to \infty}  d_{b{\mathcal{L}}}(\mu^N, \mu)=0$$
 is equivalent to  
\begin{equation}  \label {ginevra3}
\lim_{N\rightarrow\infty}\int g(w)\mu^{N}(dw)=\int g(w)\mu (dw) \,,
\end{equation}
for all bounded and continuous function $g $ on $ X$.  In the following  we denote the convergence in   \eqref {ginevra3}  by   $  \mu^N \overset{w}{\Longrightarrow}\mu $.

For  $  (q_{j},p_{j}) \in \R^{2d},j=1,..,N$,   we denote     by $ \mu^N$  the empirical measure
\begin{equation}
\mu^{N}(dx,dv):=\frac{1}{N}\sum_{j=1}^{N}\delta(q_{j}-x)\delta
(p_{j}-v)dxdv, \  \label{d.1} 
\end{equation}
 where $\delta(x-y)dx$ is the Dirac measure at
$y\in\mathbb{R}^{d}.$    
 Hence, $\mu^{N}_t$ denotes the empirical measure \eqref {d.1} when    the  $((q_{j} (t),p_{j} (t))$, $j=1,..,N$,  are the solutions of \eqref {mo1}. In this case we say that   $\mu^{N}_t$ is the empirical measure  at time $t$ associated to $w(t,w^0)$, where  $w^0= (\mathbf{q}^{0},\mathbf{p}^{0})$.  Given a
smooth function $g$ on $X\times B_{1}$ and $\mu\in \MM $ we denote
by
\begin{equation}
\mu(g)=\int_{X\times B_{1}}g(x,v)\mu(dx,dv)\  \label{d.2}%
\end{equation}
and
\[
(U\star\mu)(x)=\int_{X\times B_{1}}U(x-y)\mu(dy,du). 
\]
We have the following main results.

\begin{thm}
\label{cor1a}    Let  $w^0=(\mathbf{q}^{0},\mathbf{p}^{0}) \in({\mathcal{T}}_{D}\times B_{1})^{N}$    and  $\mu^N_{t} $, $t\ge0$,  be  the
empirical measure  associated to $w(t,w^0)$, the
solution of (\ref{mo1})  with $U$ chosen as Definition \ref{no1}.   Let     $\mu_0 \in \MM$ be such that 
\begin {equation} \lim_{N\rightarrow\infty}d_{b{\mathcal{L}}}(\mu_{0}^{N},\mu_0)=0.
\end{equation}
 Then, there exists
$\mu_{t}\in{\mathcal{M}}$ such that
\begin{equation}
\lim_{N\rightarrow\infty}d_{b{\mathcal{L}}}(\mu_{t}^{N},\mu_{t})=0,
\end{equation}
where $\mu_{t}$ is the measure solution of the following equation
\begin{equation}
\frac{\partial(\mu_{t}(g))}{\partial t}=\mu_{t}(v\cdot\nabla_{x}g)+\mu
_{t}(M\left(  \cdot,\cdot,\mu_{t}\right)  \cdot\nabla_{v}g)\ ,\forall
g\in{\mathcal{D\ }}, \label{b1a}%
\end{equation}
and for $\nu\in{\mathcal{M}},$
\begin{equation}
\label{s1}{\mathcal{T}}_{D}\times B_{1}\ni\left(  x,v\right)  \longmapsto
M(x,v,\nu):=\left(  \frac{\int_{{\mathcal{T}}_{D}\times B_{1}}U(x-y)u\nu
(dy,du)}{\int_{{\mathcal{T}}_{D}\times B_{1}}U(x-y)\nu(dy,du)}\right)
-v\in\mathbb{R}^{d}\ .
\end{equation}

\end{thm}

The next result establishes that under regularity assumptions on the initial  measure
$\mu_{0}$ and on the interaction $U$ the solution $\mu_{t}$ of \eqref{b1a} is
regular as well.

\begin{thm}
\label{th2a} Take $U$ as in Definition \ref{no1}. If $\mu_{0}(dx,dv)=f_{0}%
(x,v)dxdv,$ then $\mu_{t}(dx,dv)=f_{t}(x,v)dxdv$ and $f_{t}$ is the weak
solution of
\begin{equation}
\frac{\partial}{\partial t}f_{t}(x,v)+v\cdot\nabla_{x}f_{t}(x,v)+\nabla
_{v}\cdot\left[  M(x,v,f_{t})f_{t}(x,v)\right]  =0\ . \label{eq1a}%
\end{equation}
Furthermore, if $f_{0}\in C^{k}(X\times B_{1}),k\geq1,$ and $U \in C^{k}(X) $
then $f_{t}\in C^{k}(X\times B_{1}).$
\end{thm}

In Section 4, see Remark \ref{lun28} and Remark \ref{no2}, we will show that
Theorem \ref{cor1a} and Theorem \ref{th2a} hold also for the $\epsilon
-$\emph{\ regularized system} \eqref {mo1} when considering the configuration
space $X$ to be either $\mathbb{R}^{d}$ or ${\mathcal{T}}_{D}$ and, for any
$\nu\in{\mathcal{M}},$ $M\left(  \cdot,\cdot,\nu\right)  $ is replaced by
\begin{equation}
X\times B_{1}\ni\left(  x,v\right)  \longmapsto M_{\epsilon}(x,v,\nu):=\left(
\frac{\int_{X\times B_{1}}U(x-y)u\nu(dy,du)}{\int_{X\times B_{1}}%
U(x-y)\nu(dy,du)+\epsilon}\right)  -v\in\mathbb{R}^{d}\ . \label{s1e}%
\end{equation}
The results are shown adapting to our context the method reported in Spohn's
book  \cite[Section5]{S} (see also Neunzert   \cite{Ne} and Dobrushin  
\cite{Do}) and some classical tools of dynamical systems. 
 The main difference between the case considered here and the
one presented in   \cite{S} is that, in our case, the dependence of
$M(\cdot,\cdot,\nu)$  from $\nu$ is not linear.  We are able to  overcome this problem  when the denominator of $ M (\cdot,\cdot,\nu)$    is strictly  bigger than a positive number. 
This is the case when the $U$ in $M (\cdot,\cdot,\nu)$ is chosen as    in  Definitions \ref {no1}.
Notice that   the denominator in  $ M_\e (\cdot,\cdot,\nu)$ is always strictly bigger than $\e$.

The existence and
the uniqueness of the measure solution of equation \eqref{b1a} is given in
Theorem \ref{th1}. The existence of weak and strong solutions of (\ref{eq1a})
follows from Theorem \ref{th2}. The qualitative behaviour of the solution of
equation (\ref{eq1a}) is analyzed in Subsection 4.1. In particular, in Lemma
\ref{entro1}, we show that the Boltzmann-Vlasov entropy is strictly decreasing in time.
 
\section{Particle dynamics}

In the following we analyze the evolution of $N$ particles according equations
\eqref {mo1}. In this section $N$ is kept fixed, so we omit in the notation to
write explicitly the dependence on $N.$

\subsection{Stability}

We first notice that if the velocities of the particles at time zero are
bounded, that is, for all $i=1,..,N,\ p_{i}^{0}\in B_{r}$ for some $r>0,$ then
they will lie in $B_{r}$ for later times. In fact we have the following result:

\begin{lem}
\label{d.3}For any $i=1,..,N,$ assume that $p_{i}\left(  0\right)  \in B_{r}.$
Then, $p_{i}(t)\in B_{r},$ for all $t>0.$
\end{lem}

\begin{proof}
Assume, without loss of generality that $r=1$ and that there is a $t^{\ast}$
such that there is at least one $p_{i}(t^{\ast})$ such that $\left\vert
p_{i}(t^{\ast})\right\vert =1$ and $\left\vert p_{j}(t^{\ast})\right\vert
\leq1$ for $j\neq i.$ Then
\begin{equation}
\frac{1}{2}\frac{d}{dt}\left\vert p_{i}(t^{\ast})\right\vert ^{2}=\frac
{\sum_{j=1}^{N}U (q_{i}(t^{\ast})-q_{j}(t^{\ast}))\left[  p_{j}(t^{\ast
})-p_{i}(t^{\ast})\right]  \cdot p_{i}(t^{\ast})}{\sum_{j=1}^{N}U
(q_{i}(t^{\ast})-q_{j}(t^{\ast}))}\leq0\ .
\end{equation}

\end{proof}

\begin{rem}
\label{lu.1}The result of Lemma \ref{d.3} holds for any positive smooth
interaction $U,$ regardless of its support. In particular, it holds if $U $
does not have compact support.
\end{rem}

 Next result shows that if at time $t=0$ the particle   velocity vector  
is close to its mean velocity vector, then, at any further time $t, $ it will
always remain close to the mean initial velocity vector. Let $\Omega
\in\mathcal{L}_{Nd}$ be the operator such that%
\begin{equation}
{\mathbb{R}}^{Nd}\ni\mathbf{x\longmapsto}\Omega\mathbf{x}\in{\mathbb{R}}%
^{Nd}\ , \label{dom3}%
\end{equation}
where $\Omega\mathbf{x}$ is the vector in $\mathbb{R}^{Nd}$ whose component
are the vectors $\left(  \Omega x\right)  _{i}=\frac{1}{N}\sum_{j=1}^{N}%
x_{j}\in\mathbb{R}^{d}, \forall i=1,..,N.$ Notice that by definition $\Omega$ is
the orthogonal projector on \linebreak$\left\{  \mathbf{x}\in{\mathbb{R}}%
^{Nd}:x_{1}=\cdots=x_{N}\right\}  .$

\begin{thm}
\label{stav0} Let $w(t,w^{0})=(\mathbf{q}(t),\mathbf{p}(t))$ be the solution
of (\ref{mo1}) at time $t$ starting from $w^{0}=(\mathbf{q}^{0},\mathbf{p}%
^{0})\in{\mathbb{R}}^{2Nd}.$ Given $\epsilon>0,$ assume that $|\mathbf{p}%
^{0}-\Omega\mathbf{p}^{0}|<\epsilon.$
Then
\begin{equation}
|\mathbf{p}(t)-\Omega\mathbf{p}^{0}|\leq\epsilon,\qquad\forall t\geq0\ .
\end{equation}
\end{thm}
\begin{proof}  
 by $v_{i} (t):=p_{i}(t)-\left(  \Omega p^{0}\right)  _{i}\in{\mathbb{R}}^{d},i=1,..,N,$ then
proceed as in the proof of Lemma \ref{d.3}. 
\end{proof}  
Note that, for any $w\in{\mathbb{R}}^{2Nd},$
\begin{equation}
\operatorname{dist}\left(  w,{\mathcal{I}}\right)  =\inf_{w^{0}\in
{\mathcal{I}}}|w-w^{0}|=\inf_{\{\mathbf{p}^{0} \in\mathbb{R}^{Nd}%
:w=(\mathbf{q}^{0},\mathbf{p}^{0})\in{\mathcal{I}}\}}|\mathbf{p}%
-\mathbf{p}^{0}|=|\mathbf{p}-\Omega\mathbf{p}|\ , \label{dom2}%
\end{equation}
where $\Omega$ is the operator defined in (\ref{dom3}). From Theorem
\ref{stav0} one deduces that the invariant manifold ${\mathcal{I}}$ is stable
for the evolution (\ref{mo1}).

\begin{cor}
\label{sta0} For any $\epsilon>0$ let $B(\epsilon,{\mathcal{I}})=\left\{
w\in{\mathbb{R}}^{2Nd}:\operatorname{dist}\left(  w,{\mathcal{I}}\right)
\leq\epsilon\right\}  $ be a neighborhood of radius $\epsilon$ of
${\mathcal{I}}.$ Let $w(t,w^{0})$ be the solution of (\ref{mo1}) at time $t$
starting from $w^{0}=(\mathbf{q}^{0},\mathbf{p}^{0})\in B(\epsilon
,{\mathcal{I}}).$ Then
\begin{equation}
\operatorname{dist}\left(  w(t,w^{0}),{\mathcal{I}}\right)  \leq2
\epsilon,\qquad\forall t>0\ .
\end{equation}

\end{cor}

\begin{proof}
By \eqref{dom2} we have
\begin{equation}
\operatorname{dist}\left(  w(t,w^{0}),{\mathcal{I}}\right)  =\left\vert
\mathbf{p}(t)-\Omega\mathbf{p}(t)\right\vert \leq|\mathbf{p}(t)-\Omega
\mathbf{p}^{0}|+|\Omega\mathbf{p}(t)-\Omega\mathbf{p}^{0}|\ .
\end{equation}
By definition of $\Omega$, see (\ref{dom3}),
\begin{equation}
|\Omega\mathbf{p}(t)-\Omega\mathbf{p}^{0}|=|\Omega(\mathbf{p}(t)-\Omega
\mathbf{p}^{0})|\leq|\mathbf{p}(t)-\Omega\mathbf{p}^{0}|\ .
\end{equation}
Hence, by Theorem \ref{stav0},
\begin{equation}
\operatorname{dist}\left(  w(t,w^{0}),{\mathcal{I}}\right)  \leq
2|\mathbf{p}(t)-\Omega\mathbf{p}^{0}|\leq2\epsilon\ ,\qquad\forall t\geq0\ .
\end{equation}

\end{proof}

\subsection{Asymptotic stability}

To prove Theorem \ref{sta1a} we rewrite the non linear system (\ref{mo1}) as
follows:
\begin{equation}
\left\{
\begin{array}
[c]{l}%
\left(
\begin{array}
[c]{c}%
\frac{d\mathbf{q}(t)}{dt}\\
\frac{d\mathbf{p}(t)}{dt}%
\end{array}
\right)  =C\left(  \mathbf{q}(t)\right)  \left(
\begin{array}
[c]{c}%
\mathbf{q}(t)\\
\mathbf{p}(t)
\end{array}
\right) \\
{\mathbf{q}}(0)=\mathbf{q}^{0},{\mathbf{p}}(0)=\mathbf{p}^{0}%
\end{array}
\right.  \label{d.12a}%
\end{equation}
where%
\begin{align}
\mathbb{R}^{Nd}  &  \ni\mathbf{q}\longmapsto C\left(  \mathbf{q}\right)
:=\left(
\begin{array}
[c]{cc}%
\mathbf{0}_{Nd} & \mathbb{I}_{Nd}\\
\mathbf{0}_{Nd} & L\left(  \mathbf{q}\right)
\end{array}
\right)  \in\mathcal{L}_{2Nd}\left(  \mathbb{R}\right)  \ ,\label{d.12c}\\
L\left(  \mathbf{q}\right)   &  :=A\left(  \mathbf{q}\right)  -\mathbb{I}_{Nd}
\label{defLq}%
\end{align}
and $A(\mathbf{q})$ is the linear operator valued function so defined
\begin{align}
\mathbb{R}^{Nd}  &  \ni\mathbf{q}\longmapsto A(\mathbf{q}):=%
\begin{bmatrix}
a_{1,1}(\mathbf{q})\mathbb{I}_{d} & a_{1,2}(\mathbf{q})\mathbb{I}_{d} & \quad
& \dots & a_{1,N}(\mathbf{q})\mathbb{I}_{d}\\
a_{2,1}(\mathbf{q})\mathbb{I}_{d} & a_{2,2}(\mathbf{q})\mathbb{I}_{d} & \quad
& \dots & a_{2,N}(\mathbf{q})\mathbb{I}_{d}\\
a_{N,1}(\mathbf{q})\mathbb{I}_{d} & \quad & \dots & a_{N,N-1}(\mathbf{q}%
)\mathbb{I}_{d} & a_{N,N}(\mathbf{q})\mathbb{I}_{d}%
\end{bmatrix}
\in\mathcal{L}_{Nd}\left(  \mathbb{R}\right) \label{d.8a}\\
a_{i,j}(\mathbf{q})  &  :=\frac{U(q_{i}-q_{j})}{\sum_{k=1}^{N}U(q_{i}-q_{k}%
)}\ ,\qquad j=1,..,N,\quad i=1,..,N. \label{v1}%
\end{align}

\begin{rem}
\label{re1} Notice that for $\mathbf{q}\in{\mathbb{R}} ^{Nd}$
\begin{equation}
a_{i,j}(\mathbf{q}) = a_{i,j}(\mathbf{q}+ \Omega\mathbf{x}), \quad
\forall\mathbf{x} \in{\mathbb{R}}^{Nd}, \qquad j=1,..,N,\quad i=1,..,N
\label{sAe1}%
\end{equation} 
and
\begin{equation}
\sum_{j=1}^{N}a_{i,j}(\mathbf{q})=1. \label{sA=1}%
\end{equation}
These two properties are important when studying the spectrum of $C\left(
\mathbf{q}\right)  $ for a fixed value of $\mathbf{q}$.
\end{rem}

\subsubsection{Spectral Analysis of $C\left(  \mathbf{q}\right)  $}

Let $\mathbf{q}\in\mathbb{R}^{Nd}$ be fixed. The eigenvalues of $C\left(
\mathbf{q}\right)  $ are the roots of the characteristic equation
\begin{equation}
\operatorname{Det}\left[  C\left(  \mathbf{q}\right)  -\lambda\mathbb{I}%
_{2Nd}\right]  =(-\lambda)^{Nd}\operatorname{Det}\left[  L\left(
\mathbf{q}\right)  -\lambda\mathbb{I}_{Nd}\right]  =0 . \label{sa1}%
\end{equation}
We need then to study the spectrum of $L\left(  \mathbf{q}\right)  $ and
therefore, by (\ref{defLq}) the spectrum of $A(\mathbf{q}).$ To do this it is
convenient to introduce the tensor space ${\mathbb{R}}^{N}\otimes{\mathbb{R}%
}^{d}.$ We denote by ${\mathcal{F}}$ the isomorphism
\begin{equation}
{\mathbb{R}}^{Nd}\ni\mathbf{x}\longrightarrow\mathcal{F}\left(  \mathbf{x}%
\right)  :=\sum_{i=1}^{N}\sum_{j=1}^{d}x_{i}^{j}e_{i}\otimes e_{j}%
\in{\mathbb{R}}^{N}\otimes{\mathbb{R}}^{d}\ ,
\end{equation}
such that ${\mathcal{F}}(\mathbf{x})_{i,j}=x_{i}^{j}$, $i=1,.., N$ and
$j=1,..,d $.

To ease the notation we omit in the following to write the dependence on
$\mathbf{q}$ if no confusion arises. We therefore set $A:=A\left(
\mathbf{q}\right)  $.
One obtains immediately that $A:{\mathbb{R}}^{Nd}\longrightarrow{\mathbb{R}%
}^{Nd}$ acts on ${\mathbb{R}}^{N}\otimes{\mathbb{R}}^{d}$ as follows
\begin{equation}
\tilde{A}\otimes\mathbb{I}_{d}:{\mathbb{R}}^{N}\otimes{\mathbb{R}}%
^{d}\longrightarrow{\mathbb{R}}^{N}\otimes{\mathbb{R}}^{d}\ ,
\end{equation}
where, by ({\ref{v1}), setting }$a_{i,j}:=a_{i,j}(\mathbf{q}),$%
\begin{equation}
\tilde{A}:=%
\begin{bmatrix}
a_{1,1} & a_{1,2} & \quad & \dots & a_{1,N}\\
a_{2,1} & a_{2,2} & \quad & \dots & a_{2,N}\\
a_{N,1} & \quad & \dots & a_{N,N-1} & a_{N,N}%
\end{bmatrix}
\ . \label{d.8b}%
\end{equation}
Namely, one has that
\begin{equation}
\left(  \tilde{A}\otimes\mathbb{I}_{d}\right)  {\mathcal{F}}(\mathbf{x}%
)={\mathcal{F}}\left(  A\mathbf{x}\right)  \ . \label{defIs}%
\end{equation}
Furthermore, denoting by $\Sigma(A)\subset\mathbb{C}$ the spectrum of $A,$%
\begin{equation}
\Sigma(A)=\Sigma(\tilde{A}\otimes\mathbb{I}_{d})=\Sigma(\tilde{A}%
)\Sigma(\mathbb{I}_{d})\footnote{If $Z:=\left\{  z_{1},..,z_{n}\right\}  $ and
$W:=\left\{  w_{1},..,w_{m}\right\}  $ are two discrete subsets of
$\mathbb{C}$ we denote by
\[
ZW:=\left\{  z_{i}w_{j}\in\mathbb{C}:i=1,..,n\ ;\ j=1,..,m\right\}  .~
\]
}\ . \label{v4}%
\end{equation}
Since the only eigenvalue of $\mathbb{I}_{d}$ is $1$ with multiplicity $d,$
the problem is reduced to study the spectrum of $\tilde{A}.$ The matrix
$\tilde{A}$ is a (right) stochastic matrix, that is it has non-negative
entries and, by (\ref{sA=1}), $\sum_{j=1}^{N}a_{i,j}=1$, $\forall i=1,..,N$.
Then, if it is irreducible one can apply the Perron-Frobenius Theorem.

Recall that a matrix $D\in\mathcal{L}_{n}\left(  \mathbb{R}\right)  $ with
non-negative entries is said to be irreducible if there exists an integer $m$
such that $D^{m}$ has strictly positive entries. We have the following.
\begin{lem}
\label{v2}  Let  $\tilde
{A}(\mathbf{q})$,      $\mathbf{q}\in{\mathbb{R}}^{Nd},$   be  irreducible. Then $1$ is the maximum eigenvalue and all the
other eigenvalues $\lambda(\mathbf{q})\in\mathbb{C}$ are strictly smaller in
absolute value of $1,$ i.e. $|\lambda(\mathbf{q})|<1.$ The eigenspace
associated to the eigenvalue $1$ is one dimensional and it is generated by the eigenvector
$\eta$,   $\eta_{i}=\frac{1}{\sqrt{N}}$ for $i=1,..,N$.  There
are no other positive eigenvectors except multiples of $\eta.$
\end{lem}

\begin{proof}
Because for any $\mathbf{q} \in{\mathbb{R}}^{Nd},$ $\Vert\tilde A\left(
\mathbf{q}\right)  \Vert_{\infty}\leq\max_{i=1,..,N}\sum_{j=1}^{N}%
a_{i,j}\left(  \mathbf{q}\right)  =1,$ we have that the maximum eigenvalue is
$1$ and any other eigenvalue $\lambda(\mathbf{q})\in\mathbb{C}$ is strictly
smaller in absolute value of $1$. By Perron Frobenius Theorem the maximum
eigenvalue is simple and the associated positive eigenvector is $\eta$ with
$\eta_{i}=\frac{1}{\sqrt{N}}$ for $i=1,..,N.$
\end{proof}

It is possible to show, assuming that     $\tilde{A}\left(
\mathbf{q}\right)  $  is irreducible,    that the spectrum of  $\tilde{A}\left(
\mathbf{q}\right)  $  is indeed real, although this information  is not relevant  for the proofs of the results. 
 \begin{rem}
For any $\mathbf{q}\in\mathbb{R}^{Nd},\tilde{A}\left(
\mathbf{q}\right)  $ represents the transition matrix for the Markov chain
with state space $\mathcal{S}_{N}:=\left\{  1,...,N\right\}  .$ By (\ref{v1})
we have that $\tilde{A}\left(  \mathbf{q}\right)  $ is reversible w.r.t. the
probability distribution $\left\{  \mu_{i}\left(  \mathbf{q}\right)  \right\}
_{i\in\mathcal{S}_{N}}$ such that $\forall i\in\mathcal{S}_{N},$%
\begin{equation}
\mu_{i}\left(  \mathbf{q}\right)  :=\frac{\sum_{j=1}^{N}U\left(  q_{i}%
-q_{j}\right)  }{\sum_{i,j=1}^{N}U\left(  q_{i}-q_{j}\right)  }>0\ ,
\end{equation}
(for an account on reversible Markov chains we refer the reader to and
\cite{St}).
Let    $\mathbb{H}_{N}\left(
\mathbf{q}\right)  $ be  the  space $\mathbb{R}^{N}$   equipped with the  scalar product
\begin{equation}
\mathbb{R}^{N}\times\mathbb{R}^{N}\ni\left(  f,g\right)  \longmapsto
\left\langle f,g\right\rangle _{\mathbf{q}}:=\sum_{i\in\mathcal{V}_{N}}\mu
_{i}\left(  \mathbf{q}\right)  g_{i}f_{i}\in\mathbb{R}\ . \label{H(q)}%
\end{equation}
It is easy to verify that  $\tilde{A}\left(  \mathbf{q}\right)  $ is selfadjoint on $\mathbb{H}%
_{N}\left(  \mathbf{q}\right)$, hence  the  eigenvalues of    $\tilde{A}\left(  \mathbf{q}\right)  $  are real.
\end{rem}

\begin{lem}
\label{v3} For any $\mathbf{q}\in{\mathbb{R}}^{Nd},$ such that $\tilde
{A}(\mathbf{q})$ is irreducible, let $A(\mathbf{q})$ be the matrix as in
(\ref{d.8a}). We have that $1\in\Sigma(A(\mathbf{q}))$ is the maximum
eigenvalue. The associated eigenspace is the $d$-dimensional manifold
$\{\mathbf{p}\in{\mathbb{R}}^{Nd}:p_{i}=v,\ i=1,..,N; v \in{\mathbb{R}}%
^{d}\}.$ All the other eigenvalues $\lambda(\mathbf{q})\in\Sigma
(A(\mathbf{q}))$ are such that $|\lambda(\mathbf{q})|<1.$
\end{lem}
\begin{proof}
It is an immediate consequence of (\ref{v4}) and Lemma \ref{v2}.
\end{proof}
 We have finally the following result.

\begin{thm}
\label{v5} For any $\mathbf{q}\in{\mathbb{R}}^{Nd},$ such that $\tilde
{A}(\mathbf{q})$ is irreducible, let $C(\mathbf{q})$ be defined in
(\ref{d.12c}). We have that $0\in\Sigma(C(\mathbf{q})).$ The $(N+1)d$
dimensional manifold ${\mathcal{I}}$ defined in (\ref{v8}) is the eigenspace
associated to the eigenvalue $0.$ All the other eigenvalues of $C(\mathbf{q})$
have real part strictly negative.
\end{thm}

\begin{proof}
From (\ref{sa1}) and Lemma \ref{v3} we deduce that $0\in\Sigma(C(\mathbf{q}))
$ and all other eigenvalues have real part strictly negative. It is immediate
to see that the algebraic multiplicity of $0$ is $Nd+d.$ The $\left(
N+1\right)  d$-dimensional manifold ${\mathcal{I}}$ defined in (\ref{v8}) is
the associated eigenspace. Namely, if $w\in{\mathcal{I}}$ then $C\left(
\mathbf{q}\right)  w\in{\mathcal{I}}.$ From this one deduces that
${\mathcal{I}}$ is an eigenspace for the matrix $C(\mathbf{q}).$ Moreover,
since the kernel of $C^{2}\left(  \mathbf{q}\right)  $ is ${\mathcal{I}},$ we
get that ${\mathcal{I}}$ is the eigenspace associated to the eigenvalue $0.$
\end{proof}

We denote by $\alpha(\mathbf{q})$ the spectral gap of the matrix
$C(\mathbf{q}),$ that is
\begin{equation}
\alpha(\mathbf{q}):=\min\left\{  \left\vert \text{Re}(\lambda(\mathbf{q}%
))\right\vert :\lambda(\mathbf{q})\in\Sigma(C(\mathbf{q})),\ \text{Re}%
(\lambda(\mathbf{q}))<0\right\}  \ .
\end{equation}
Let $\mathbf{q}\in{\mathbb{R}}^{Nd}$ such that $\tilde{A}(\mathbf{q})$ is
irreducible. By Theorem \ref{v5}, ${\mathcal{I}}$ is the eigenspace associated
to the $0$ eigenvalue of $C(\mathbf{q})$ for any $\mathbf{q}.$ We can
therefore decompose ${\mathbb{R}}^{2Nd}$ as follows:
\begin{equation}
{\mathbb{R}}^{2Nd}=W(\mathbf{q})\oplus{\mathcal{I}}%
\end{equation}
in such a way that $W(\mathbf{q})$ and ${\mathcal{I}}$ are eigenspaces of
$C(\mathbf{q})$ and denote by $\Pi(\mathbf{q})$ the projection operator
\begin{equation}
\Pi(\mathbf{q}):{\mathbb{R}}^{2Nd}\rightarrow W(\mathbf{q})\ .
\end{equation}

\subsubsection{Asymptotic Analysis}

Let $w^{0}=(\mathbf{q}^{0},\mathbf{p}^{0})\in{\mathcal{I}}$ be such that
$\tilde{A}(\mathbf{q}^{0})$ is irreducible and let us set, for any $r>0$ and
$\epsilon>0,$
\begin{equation}
\mathcal{\tilde{B}}(r,\epsilon,w^{0}):=\{w=(\mathbf{q},\mathbf{p}%
)\in{\mathbb{R}}^{2Nd}:|\mathbf{q}-\mathbf{q}^{0}|\leq r\ ;\ |\mathbf{p}%
-\mathbf{p}^{0}|\leq\epsilon\}\ . \label{mar3}%
\end{equation}
Denote by $r_{0}$ the biggest value of $r$ such that, for any $w=(\mathbf{q}%
,\mathbf{p})\in\mathcal{\tilde{B}}(r_{0},\epsilon,w^{0}),\tilde{A}%
(\mathbf{q})$ is still irreducible and $\alpha(\mathbf{q})\geq\frac{1}%
{2}\alpha(\mathbf{q}^{0})$. We set
\begin{equation}
\mathcal{B}(r_{0},\epsilon,w^{0}):=\left\{  w=(\mathbf{q},\mathbf{p}%
)\in\mathcal{\tilde{B}}(r_{0},\epsilon,w^{0}):\alpha(\mathbf{q})\geq\frac
{1}{2}\alpha(\mathbf{q}^{0})\right\}  \ . \label{dom6}%
\end{equation}
The existence of $r_{0}$ is granted since, by assumption, $\tilde
{A}(\mathbf{q}^{0})$ is irreducible and $U$ is smooth. To apply the spectral
results obtained for $C(\mathbf{q})$ ($\mathbf{q}$ fixed) to the nonlinear
system \eqref {d.12a} we write
\begin{equation}
C(\mathbf{q}(t))=C(\mathbf{q}(0))+\Gamma(\mathbf{q}(t))\ ,
\end{equation}
where
\begin{equation}
\Gamma(\mathbf{q}(t)):=\left(
\begin{array}
[c]{cc}%
\mathbf{0}_{Nd} & \mathbf{0}_{Nd}\\
\mathbf{0}_{Nd} & B\left(  \mathbf{q}\left(  t\right)  \right)
\end{array}
\right)  \ , \label{defK}%
\end{equation}
and
\begin{equation}
\label{roma2}B(\mathbf{q}(t)):=A(\mathbf{q}(t))-A(\mathbf{q}(0))\ .
\end{equation}
Next we estimate the norm of $B(\mathbf{q}(t)).$

\begin{lem}
\label{gi1} Let $(\mathbf{q}(t),\mathbf{p}(t))$ be the solution of
(\ref{d.12a}) starting from the initial data $(\mathbf{q}^{0},\mathbf{p}%
^{0}).$ We have
\begin{equation}
\left\Vert B(\mathbf{q}(t))\right\Vert \leq2N\frac{\sup_{x\in\mathbb{R}^{d}%
}\left\vert \nabla U(x)\right\vert }{U \left(  0\right)  + (N-1)
\eta(\mathbf{q}(t), \mathbf{q}^{0}) }\sup_{i,k\in\{1,..,N\}}\left\vert
-(q_{i}^{0}-q_{k}^{0})+q_{i}(t)-q_{k}(t)\right\vert \ ,
\end{equation}
where $\eta(\mathbf{q}(t), \mathbf{q}^{0}) \ge0 $ is defined in \eqref {roma1}.
\end{lem}
We defer the proof of this result to the appendix. In the proof of
\eqref   {sta1a} we will use Lemma \ref{gi1}  taking 
$\eta(\mathbf{q}(t),\mathbf{q}^{0})=0$.

\vskip1.cm

\noindent\textbf{Proof of Theorem \eqref {sta1a}:} For any $s>0,$ we define
\begin{equation}
{\mathcal{\tilde{Q}}}(s,w^{0}):=\{w=(\mathbf{q},\mathbf{p})\in{\mathbb{R}%
}^{2Nd}:|[\mathbb{I}_{Nd}-\Omega]\left(  \mathbf{q}-\mathbf{q}^{0}\right)
|\leq s\}\ , \label{mar4}%
\end{equation}
where $\Omega$ is the operator defined in (\ref{dom3}). Denote by $s_{0}$ the
largest value of $s$ such that, for any $w=(\mathbf{q},\mathbf{p}%
)\in{\mathcal{\tilde{Q}}}(s_{0},w^{0}),\tilde{A}(\mathbf{q})$ is still
irreducible and $\alpha(\mathbf{q})\geq\frac{1}{4}\alpha(\mathbf{q}^{0})$.
Such a value $s_{0}$ exists since $\tilde{A}(\mathbf{q}^{0})$ is irreducible
and $U$ is smooth. Let us set
\begin{equation}
\mathcal{Q}(s_{0},w^{0}):=\left\{  w=(\mathbf{q},\mathbf{p})\in
{\mathcal{\tilde{Q}}}(s_{0},w^{0}):\alpha(\mathbf{q})\geq\frac{1}{4}%
\alpha(\mathbf{q}^{0})\right\}  \ . \label{dom5}%
\end{equation}
We have that
\begin{equation}
\mathcal{B}(r_{0},\epsilon,w^{0})\subset{\mathcal{Q}}(s_{0},w^{0}%
)\ ,\qquad\forall\epsilon>0\ .
\end{equation}
Namely we have that $s_{0}\geq r_{0}$ since requirement (\ref{dom6}) is
stronger than (\ref{dom5}) and
\begin{equation}
|[\mathbb{I}_{Nd}-\Omega](\mathbf{q}-\mathbf{q}^{0})|\leq|\mathbf{q}%
-\mathbf{q}^{0}|\leq r_{0}\ .
\end{equation}
Let $w(t,w^{1})=(\mathbf{q}(t,w^{1}),\mathbf{p}(t,w^{1}))$ be the solution of
system (\ref{d.12a}) starting from an initial datum $w^{1}\in\mathcal{B}%
(r_{0},\epsilon,w^{0})$ and let $t^{\ast}(w_{1})>0$ be the first exit time of
$w(t,w^{1})$ from ${\mathcal{Q}}(s_{0},w^{0}).$ If $w(t,w^{1})\in{\mathcal{Q}%
}(s_{0},w^{0})$ for all $t\geq0,$ then we set $t^{\ast}(w^{1})=\infty.$ Next
we analyze the solution for $t<t^{\ast}(w^{1})$ and we will show that
$t^{\ast}(w^{1})=\infty$ for any initial datum $w^{1}\in\mathcal{B}%
(r_{0},\epsilon,w^{0}),$ provided that $\epsilon$ in (\ref{mar3}) is suitably
chosen. Let us define
\begin{align}
\xi(t)  &  :=\Pi(\mathbf{q}(t,w^{1}))w(t,w^{1})\ ,\\
\chi(t)  &  :=\left(  \mathbb{I}_{2Nd}-\Pi(\mathbf{q}(t,w^{1})\right)
w(t,w^{1})\ ,\quad t<t^{\ast}(w^{1})\ .
\end{align}
By construction $\chi(t)\in{\mathcal{I}},\xi(t)\in W(\mathbf{q}(t,w^{1})).$ We
then have
\begin{align}
\frac{d}{dt}\xi(t)  &  =\left(  \frac{d}{dt}\Pi(\mathbf{q}(t))\right)
w(t,w^{1})+\Pi(\mathbf{q}(t))\frac{d}{dt}w(t,w^{1})\label{sa9}\\
&  =\left(  \frac{d}{dt}\Pi(\mathbf{q}(t))\right)  w(t,w^{1})+\Pi
(\mathbf{q}(t))C(\mathbf{q}(t))w(t,w^{1})\ .\nonumber
\end{align}
Taking into account that $w(t,w^{1})=\xi(t)+\chi(t)$ we get
\begin{equation}
\frac{d}{dt}\xi(t)=\left(  \frac{d}{dt}\Pi(\mathbf{q}(t))\right)
\xi(t)+\left(  \frac{d}{dt}\Pi(\mathbf{q}(t))\right)  \chi(t)+\Pi
(\mathbf{q}(t))C(\mathbf{q}(t))\xi(t)\ . \label{sa9a}%
\end{equation}
Since for any given $w\in{\mathcal{I}},$ by the definition $\Pi(\mathbf{q}%
(t)),$ we have $\frac{d}{dt}\Pi(\mathbf{q}(t))w=0$ and $C(\mathbf{q}(t))$ and
$\Pi(\mathbf{q}(t))$ commute, we obtain
\begin{equation}
\frac{d}{dt}\xi(t)=\left(  \frac{d}{dt}\Pi(\mathbf{q}(t))\right)
\xi(t)+C(\mathbf{q}(t))\xi(t)\ . \label{lun2}%
\end{equation}
Setting
\begin{equation}
C(\mathbf{q}(t))=C(\mathbf{q}(0))+\Gamma(\mathbf{q}(t))\ ,
\end{equation}
where $\Gamma(\mathbf{q}(t))$ is defined in \eqref {defK}, we get
\begin{equation}
\frac{d}{dt}\xi(t)=\left(  \frac{d}{dt}\Pi(\mathbf{q}(t))\right)
\xi(t)+C(\mathbf{q}(0))\xi(t)+\Gamma(\mathbf{q}(t))\xi(t)\ . \label{lun2a}%
\end{equation}
By the formula of variation of constants:
\begin{equation}
\xi(t)=e^{C(\mathbf{q}(0))t}\xi(0)+\int_{0}^{t}e^{C(\mathbf{q}(0))(t-s)}%
\left\{  \left(  \frac{d}{ds}\Pi(\mathbf{q}(s))\right)  \xi(s)+\Gamma
(\mathbf{q}(s))\xi(s)\right\}  ds. \label{lun2b}%
\end{equation}
Performing the exponential of the matrix $C(\mathbf{q}(0))$ one needs to take
into account that, because of the possible presence of Jordan blocks, powers
of $t$ might appear. We control such terms paying $e^{-\frac{1}{2}%
\alpha(\mathbf{q}(0))t}$ and multiplying the remaining exponential by a
constant $D(C(\mathbf{q}(0)))$ which depends only on $C(\mathbf{q}(0)).$ Since
$\mathbf{q}(0)\in\mathbf{B}_{r_{0}}\left(  \mathbf{q}^{0}\right)  ,$ which is
a compact set in ${\mathbb{R}}^{Nd},$ we denote by $D_{0}:=\sup_{\mathbf{q}%
\in\mathbf{B}_{r_{0}}\left(  \mathbf{q}^{0}\right)  }D(C(\mathbf{q})),$ which
depends only on $C(\mathbf{q}^{0})$ and $r_{0}.$ Therefore, we get
\begin{equation}
|\xi(t)|\leq D_{0}e^{-\frac{1}{2}\alpha(\mathbf{q}(0))t}|\xi(0)|+D_{0}\int
_{0}^{t}e^{-\frac{1}{2}\alpha(\mathbf{q}(0))(t-s)}\left\{  |\left(  \frac
{d}{ds}\Pi(\mathbf{q}(s))\right)  \xi(s)|+|\Gamma(\mathbf{q}(s))\xi
(s)|\right\}  ds\ . \label{lun2c}%
\end{equation}
Next we estimate $\left\Vert \frac{d}{dt}\Pi(\mathbf{q}(t))\right\Vert .$ Let
$\Pi(\mathbf{q}(t))=\left\{  \pi_{i,j}(\mathbf{q}(t)) \otimes \mathbb{I}_{d} \right\}  _{i,j=1,..,,N}%
,$ we then have
\begin{align}
\frac{d}{dt}\Pi(\mathbf{q}(t))  &  =\left\{  \nabla_{\mathbf{q}(t)}\pi
_{i,j}(\mathbf{q}(t))\cdot\mathbf{p}(t) \otimes \mathbb{I}_{d}\right\}_{i,j=1,..,,N}\\
&  =\left\{  \nabla_{\mathbf{q}(t)}\pi_{i,j}(\mathbf{q}(t))\cdot\left[
\mathbf{p}(t)-\Omega\mathbf{p}(t)\right] \otimes \mathbb{I}_{d} \right\}_{i,j=1,..,,N}\ .\nonumber
\end{align}
The last equality holds since, by \eqref {sAe1}, $\Pi(\mathbf{q}%
(t))=\Pi(\mathbf{q}(0)),\forall t\in{\mathbb{R}},$ when $\mathbf{q}(t)$ is the
evolution given by the flow on the invariant manifold, i.e. when
$\mathbf{p}(t)=\Omega\mathbf{p}(t).$ We get by Corollary \ref{sta0}
\begin{align}
\left\Vert \frac{d}{dt}\Pi(\mathbf{q}(t))\right\Vert  &  \leq\sup
_{\{\mathbf{q}\in\mathbb{R}^{Nd}\ :\ w=(\mathbf{q},\mathbf{p})\in{\mathcal{Q}%
}(s_{0},w^{0})\}}\sup_{i=1,..,N}\sum_{j=1}^{N}|\nabla_{{q}}\pi_{i,j}(\mathbf{q}%
)|\left\vert \mathbf{p}(t)-\Omega\mathbf{p}(t)\right\vert \label{mar1}\\
&  \leq D^{\prime}(s_{0})\epsilon\ ,\qquad\forall t\in\lbrack0,t^{\ast}%
(w^{1}))\ ,\nonumber
\end{align}
where $D^{\prime}(s_{0})>0.$ Furthermore, by (\ref{defK}) and Lemma \ref{gi1},
we have
\begin{equation}
\left\Vert \Gamma(\mathbf{q}(t))\right\Vert =\left\Vert B(\mathbf{q}%
(t))\right\Vert \leq2\frac{N}{U\left(  0\right)  }\sup_{x\in\mathbb{R}^{d}%
}\left\vert \nabla U(x)\right\vert \max_{1\leq i,k\leq N}\left\vert
-(q_{i}^{1}-q_{k}^{1})+q_{i}(t)-q_{k}(t)\right\vert
\end{equation}
and, by Theorem \ref{stav0}, for $i,k=1,..,N,\ $%
\begin{align}
\left\vert q_{i}(t)-q_{k}(t)-(q_{i}^{1}-q_{k}^{1})\right\vert  &  =\int
_{0}^{t}\left\vert p_{i}(s^{\prime})-p_{k}(s^{\prime})\right\vert ds^{\prime
}\\
&  =\int_{0}^{t}\left\vert p_{i}(s^{\prime})-p_{i}^{0}+p_{k}^{0}%
-p_{k}(s^{\prime})\right\vert ds^{\prime}\leq2\epsilon t,\nonumber
\end{align}
where $p_{i}^{0}=p_{k}^{0}$ since $(\mathbf{q}^{0},\mathbf{p}^{0}%
)\in{\mathcal{I}}.$ Thus, setting $D_{1}:=2\frac{\sup_{x\in\mathbb{R}^{d}%
}\left\vert \nabla U(x)\right\vert }{U\left(  0\right)  },\forall t\in
\lbrack0,t^{\ast}(w^{1}))$ we obtain
\begin{equation}
|\xi(t)|\leq D_{0}e^{-\frac{1}{2}\alpha(\mathbf{q}(0))t}|\xi(0)|+D_{0}%
\epsilon\int_{0}^{t}e^{-\frac{1}{2}\alpha(\mathbf{q}(0))(t-s)}\left\{  \left[
D^{\prime}\left(  s_{0}\right)  +2ND_{1}s\right]  |\xi(s)|\right\}  ds\ .
\label{lun3}%
\end{equation}
Given $K\geq\max\{D^{\prime}\left(  s_{0}\right)  ,2D_{1}\},$ take
$T\in\left(  0,t^{\ast}(w^{1})\right)  $. A suitable choice of $T$ will be
done later. Then,
\begin{equation}
|\xi(t)|\leq D_{0}e^{-\frac{1}{2}\alpha(\mathbf{q}(0))t}|\xi(0)|+\epsilon
D_{0}K\left\{  1+TN\right\}  \int_{0}^{t}e^{-\frac{1}{2}\alpha(\mathbf{q}%
(0))(t-s)}|\xi(s)|ds,\qquad\forall t\in\lbrack0,T]. \label{ven5}%
\end{equation}

By the Gronwall's inequality  we get
\begin{equation}
|\xi(t)|\leq D_{0}|\xi(0)|e^{-t\left[  \frac{1}{2}\alpha(\mathbf{q}%
(0))-\epsilon\delta\right]  }\leq D_{0}|\xi(0)|e^{-t\left[  \frac{1}{8}%
\alpha(\mathbf{q}^{0})-\epsilon\delta\right]  }\ ,\quad\forall t\in
\lbrack0,T]\ , \label{lun4}%
\end{equation}
where we made use of (\ref{dom5}) and set $\delta:=D_{0}K\left\{
1+NT\right\}  .$ Let us choose $\epsilon$ such that
\begin{equation}
\frac{1}{16}\alpha(\mathbf{q}^{0})\geq\epsilon D_{0}K\{1+NT\}\ . \label{lune1}%
\end{equation}
Then,
\begin{equation}
|\xi(t)|\leq D_{0}|\xi(0)|e^{-t\frac{1}{16}\alpha(\mathbf{q}^{0})}%
\ ,\quad\forall t\in\lbrack0,T]\ . \label{dom1}%
\end{equation}

Since
\begin{align}
\operatorname{dist}\left(  w(t,w^{1}),{\mathcal{I}}\right)   &  =\inf
_{\tilde{w}\in{\mathcal{I}}}|w(t,w^{1})-\tilde{w}|=\inf_{\tilde{w}%
\in{\mathcal{I}}}|\chi(t)+\xi(t)-\tilde{w}|\label{distI}\\
&  \leq|\xi(t)|+\inf_{\tilde{w}\in{\mathcal{I}}}|\chi(t)-\tilde{w}%
|=|\xi(t)|\ ,\nonumber
\end{align}
we have
\begin{equation}
\operatorname{dist}\left(  w(t,w^{1}),{\mathcal{I}}\right)  \leq D_{0}%
|\xi(0)|e^{-t\frac{1}{16}\alpha(\mathbf{q}^{0})}\ ,\quad\forall t\in
\lbrack0,T]\ .
\end{equation}
Then, recalling that $\operatorname{dist}\left(  w(t,w^{1}),{\mathcal{I}%
}\right)  =|\mathbf{p}(t)-\Omega\mathbf{p}(t)|,$ and, since for $w^{1}%
\in\mathcal{B}(r_{0},\epsilon,w^{0}),$%
\begin{equation}
|\xi(0)|\leq D(s_{0})\operatorname{dist}\left(  w^{1},{\mathcal{I}}\right)
\leq D(s_{0})\epsilon\ , \label{ven6}%
\end{equation}
we have
\begin{equation}
|\mathbf{p}(t)-\Omega\mathbf{p}(t)|\leq D_{0}D(s_{0})\epsilon e^{-\frac{1}%
{16}\alpha(\mathbf{q}^{0})t}\leq\epsilon(D_{0}D(s_{0})\vee1)e^{-\frac{1}%
{16}\alpha(\mathbf{q}^{0})t}\ ,\quad\forall t\in\lbrack0,T]\ . \label{lune3}%
\end{equation}
\ From this we get
\begin{align}
|[\mathbb{I}_{Nd}-\Omega]\left(  \mathbf{q}(t)-\mathbf{q}^{0}\right)  |  &
\leq|[\mathbb{I}_{Nd}-\Omega]\left(  \mathbf{q}(0)-\mathbf{q}^{0}\right)
|+\int_{0}^{t}|\left[  \mathbb{I}_{Nd}-\Omega\right]  \mathbf{p}(s)|ds\\
&  \leq|\mathbf{q}(0)-\mathbf{q}^{0}|+(D_{0}D(s_{0})\vee1)\frac{16}%
{\alpha(\mathbf{q}^{0})}\epsilon(1-e^{-\frac{1}{16}\alpha(\mathbf{q}^{0}%
)t})\nonumber\\
&  \leq|\mathbf{q}(0)-\mathbf{q}^{0}|+(D_{0}D(s_{0})\vee1))\frac{16}%
{\alpha(\mathbf{q}^{0})}\epsilon\ .\nonumber
\end{align}
Let us choose $\epsilon$ such that
\begin{equation}
r_{0}+(D_{0}D(s_{0})\vee1))\frac{16}{\alpha(\mathbf{q}^{0})}\epsilon\leq
\frac{1}{2}s_{0}\ , \label{lune2}%
\end{equation}
and denote this chosen value by $\tilde{\epsilon}_{1}.$ Now we first choose
$T$ such that
\begin{equation}
(D_{0}D(s_{0})\vee1)e^{-\frac{1}{16}\alpha(\mathbf{q}^{0})T}=\frac{1}{2}%
\wedge\frac{1}{2D_{0}} \label{lune5}%
\end{equation}
and denote this chosen value by $T_{0},$ then we choose $\tilde{\epsilon}_{2}$
in such a way that (\ref{lune1}) holds with $T$ replaced by $T_{0}.$ We then
set
\begin{equation}
\epsilon_{0}:=\min\left\{  \tilde{\epsilon}_{1},\tilde{\epsilon}_{2}\right\}
\ . \label{lune7}%
\end{equation}
Notice that, by (\ref{lune2}),
\begin{equation}
\tilde{\epsilon}_{1}\frac{16}{\alpha(\mathbf{q}^{0})}\leq r_{0}+(D_{0}%
D(s_{0})\vee1))\frac{16}{\alpha(\mathbf{q}^{0})}\epsilon\leq\frac{1}{2}s_{0}%
\end{equation}
so, since $\alpha\left(  \mathbf{q}^{0}\right)  <1,$%
\begin{equation}
\epsilon_{0}\leq\tilde{\epsilon}_{1}\leq\frac{1}{32}s_{0}\ .
\end{equation}
We remark that the choice of $T_{0}$ and $\epsilon_{0}$ depends on $w^{0}%
\in{\mathcal{I}}$. Therefore, at time $T_{0}$ we have
\begin{equation}
|[\mathbb{I}_{Nd}-\Omega]\left(  \mathbf{q}(T_{0})-\mathbf{q}^{0}\right)
|\leq\frac{1}{2}s_{0}%
\end{equation}
and
\begin{equation}
\operatorname{dist}\left(  w(T_{0},w^{1}),{\mathcal{I}}\right)  =|\mathbf{p}%
(T_{0})-\Omega\mathbf{p}(T_{0})|\leq\frac{\epsilon_{0}}{2}\ .
\end{equation}
We can then repeat the previous argument for the solution of the system
(\ref{mo1}) starting at time $T_{0}$ from the initial datum $(\mathbf{q}%
(T_{0}),\mathbf{p}(T_{0})).$ We need to recall that $\alpha(\mathbf{q}%
(T_{0}))\geq\frac{1}{4}\alpha(\mathbf{q}^{0}).$ In a similar way we can show
that that for $t\in\lbrack T_{0},2T_{0}],$%
\begin{equation}
|\mathbf{p}(t)-\Omega\mathbf{p}(t)|\leq D_{0}|\xi(T_{0})|e^{-(t-T_{0})\frac
{1}{16}\alpha(\mathbf{q}^{0})}\ ,\quad\forall t\in\lbrack T_{0},2T_{0}]\ .
\label{lune3a}%
\end{equation}
Therefore, by (\ref{lune5}), we have
\begin{equation}
\operatorname{dist}\left(  w(2T_{0},w(T_{0})),{\mathcal{I}}\right)
=|\mathbf{p}(2T_{0})-\Omega\mathbf{p}(2T_{0})|\leq\frac{\epsilon_{0}}%
{2^{2}(D_{0}D(s_{0})\vee1)}\leq\frac{\epsilon_{0}}{2^{2}}\ ,
\end{equation}
and, by \eqref {lune3a},
\begin{align}
|[\mathbb{I}_{Nd}-\Omega]\left(  \mathbf{q}(t)-\mathbf{q}^{0}\right)  |  &
\leq|[\mathbb{I}_{Nd}-\Omega]\left(  \mathbf{q}(T_{0})-\mathbf{q}^{0}\right)
|+\int_{T_{0}}^{t}|[\mathbb{I}_{Nd}-\Omega]\mathbf{p}(s)|ds\\
&  \leq|[\mathbb{I}_{Nd}-\Omega]\left(  \mathbf{q}(T_{0})-\mathbf{q}%
^{0}\right)  |+D_{0}|\xi(T_{0})|\int_{T_{0}}^{t}e^{-(s-T_{0})\frac{1}%
{16}\alpha(\mathbf{q}^{0})}ds\nonumber\\
&  \leq\frac{1}{2}s_{0}+\frac{1}{4}s_{0}\ ,\nonumber
\end{align}
the last inequality being a consequence of \eqref {lune2} and (\ref{lune5}).
Thus, at time $T_{1}=2T_{0}$%
\begin{equation}
|[\mathbb{I}_{Nd}-\Omega]\left(  \mathbf{q}(T_{1})-\mathbf{q}^{0}\right)
|\leq\frac{1}{2}s_{0}+\frac{s_{0}}{4}\ .
\end{equation}
Hence, we have that $(\mathbf{q}(T_{1}),\mathbf{p}(T_{1}))\in{\mathcal{Q}%
}(s_{0},w^{0}).$ Iterating this procedure $m$ times we get%
\begin{equation}
\operatorname{dist}\left(  w(T_{m},w^{1}),{\mathcal{I}}\right)  =|\mathbf{p}%
(T_{m})-\Omega\mathbf{p}(T_{m})|\leq\frac{\epsilon_{0}}{2^{m+1}}\ ,
\end{equation}
and
\begin{equation}
|[\mathbb{I}_{Nd}-\Omega]\left(  \mathbf{q}(T_{m})-\mathbf{q}^{0}\right)
|\leq s_{0}\sum_{k=0}^{m}\frac{1}{2^{k+1}}\ .
\end{equation}
Since $\sum_{k\geq1}\frac{1}{2^{m}}=\frac{1}{2}$ we obtain the thesis of the
theorem. \qed

\section{Kinetic limit: Vlasov type equation}

We  study system \eqref{mo1} when the number of particles $N$
goes to infinity and   derive the kinetic equation for the density
$f_{t}(x,v)$ of particles at $x$ with velocity $v$ at time $t.$
The heuristic  argument goes as following.
 Let    $\mu^{N}_t$ be  the empirical measure, see \eqref {d.1},   at time $t$ associated to $w(t,w^0)$, solution of the system \eqref{mo1}, where  $w^0= (\mathbf{q}^{0},\mathbf{p}^{0})$,  $\Vert p_{j}^{0}\Vert\leq1,j=1,..,N$. By Lemma \ref{d.3}
and Remark \ref{lu.1} $\mu_{t}^{N}$ has support on ${\mathbb{R}}^{d}\times
B_{1}$, for all $t \ge 0$. Writing the second equation of (\ref{mo1}) in term of $\mu_{t}^{N}$ we
get
\begin{equation}   \begin {split} \label{d2a}
\frac{dp_{i}(t)}{dt}  &  =\frac{\int_{{\mathbb{R}}^{d}\times B_{1}}U(q_{i}(t)-y)\left(
u-p_{i}\left(  t\right)  \right)  \mu_{t}^{N}(dy,du)}{\int_{{\mathbb{R}}%
^{d}\times B_{1}}U(q_{i}(t)-y)\mu_{t}^{N}(dy,du)} 
\cr &  =:M(q_{i}(t),p_{i}(t),\mu_{t}^{N}).
\end{split} \end {equation}
Therefore, the evolution of $\mu_{t}^{N}$ is given by
\begin{equation}
\frac{\partial(\mu_{t}^{N}(g))}{\partial t}=\mu_{t}^{N}(v\cdot\nabla_{x}%
g)+\mu_{t}^{N}(M\left(  \cdot,\cdot,\mu_{t}^{N}\right)  \cdot\nabla_{v}g)\ ,
\label{mu0a}%
\end{equation}
where $g$ is a smooth test function. In the equation (\ref{mu0a}), $N$ is
fixed. To study the limit as $N\rightarrow\infty$ we assume that at $t=0$ there exists $\mu_{0} \in \MM$ such that 
\begin{equation}
\mu_{0}^{N} \overset{w}{\Longrightarrow}\mu_{0}.   \label{mu0}%
\end{equation}
 We want to show that if (\ref{mu0}) holds at time $t=0,$ then
\begin{equation}
\mu_{t}^{N}\overset{w}{\Longrightarrow}\mu_{t}, \label{mut}%
\end{equation}
where $\mu_{t}$ is the measure solution of the following equation
\begin{equation}
\frac{\partial(\mu_{t}(g))}{\partial t}=\mu_{t}(v\cdot\nabla_{x}g)+\mu
_{t}(M\left(  \cdot,\cdot,\mu_{t}\right)  \cdot\nabla_{v}g)\ , \label{b1}%
\end{equation}
which is the  formal limit of (\ref{mu0a}).   To prove  it rigorously one needs  to have $M\left(  \cdot,\cdot, \nu\right)$  well defined  and Lipschitz continuous in $(x,v)$  for all $\nu \in \MM$.   But, already at $N$ finite,   the denominator of \eqref {d2a}  is  equal to zero   when the supports of $U$ and  $\mu_{t}^{N}$ are disjoint.
To overcome these problems we consider two classes
of interaction $U.$ The first one is the class of interactions in Definition \ref {no1}.  In this case we define
$M(x,v,\nu)$ as in \eqref  {s1}. The second one is the class of smooth
interactions $U$ with compact support. In this case  we fix $\epsilon>0$ and 
define $M_{\epsilon}(x,v,\nu)$ as in \eqref {s1e},  the $\e-$ regularized system. Hence, in the case $U$ has
compact support, we modify the interaction term in such a way that when
$\int_{X\times B_{1}}U(x-y)\nu(dy,du)=0$ then $M_{\epsilon}(x,v,\nu)=0,$ when
$\int_{X\times B_{1}}U(x-y)\nu(dy,du)>\epsilon$ then $M_{\epsilon}%
(x,v,\nu)=M(x,v,\nu)+O(\epsilon),$ when $\epsilon>\int_{X\times B_{1}%
}U(x-y)\nu(dy,du)>0$ then $M_{\epsilon}(x,v,\nu)$ is a large perturbation of
$M(x,v,\nu).$  
It is easy to see that for any measure $\nu$ on $X \times B_{1}$ we have
\begin{equation}
\sup_{(x,v)\in X \times B_{1}}\left\vert M(x,v,\nu) \right\vert \leq
2\ ,\qquad\sup_{(x,v)\in X \times B_{1}}\left\vert M_{\epsilon}(x,v,\nu
)\right\vert \leq2\ . \label{feb1}%
\end{equation}
The Lipschitz continuity of $M(\cdot,\cdot,\nu)$ with respect to $v$ follows
from the linearity of $M(\cdot,\cdot,\nu)$ as a function of $v.$ The Lipschitz
continuity of $M(\cdot,\cdot,\nu)$ with respect to $x$ does not hold in
general even if one takes smooth interactions $U$,     due to  the
 presence of the denominator in $M(\cdot,v,\nu).$  In  Lemma \eqref {le2}, Lemma \ref {le3} and Lemma \ref {le5} we show for three different type of interactions, respectively  for $U$ defined  trough $\tilde U$ as in \eqref {ginevra1} and  \eqref {ginevra2} and for the  $\e-$ regularised system,
that  $M(\cdot,v,\nu)$  is Lipschiz for all $v$ and $\nu$.
We denote by
\begin {equation}   \label {gi4}  A(\cdot,\nu) := \left(  \frac{\int_{{\mathcal{T}}_{D}\times B_{1}}U(x-y)u\nu
(dy,du)}{\int_{{\mathcal{T}}_{D}\times B_{1}}U(x-y)\nu(dy,du)}\right).
\end {equation}

\nada {The denominator of the
derivative in $x^{i},i=1,..,d,$ of any component of the vector $A(\cdot,\nu)$
might be very small while the numerator, because of the presence of the
derivative of $U,$ might be not of the same order. Hence, one can certainly
control the gradient of $A(\cdot,\nu)$ if the gradient of $U$ is of the same
order of $U.$ We have then the following lemma.}

\begin{lem}
\label{le2}  Let $U$ be the   interaction defined on $\mathcal{T}_{D}$ through
the periodization of $\tilde{U}$ as  defined in \eqref {ginevra1}.
 For $\nu\in{\mathcal{M}}$ and $A(\cdot,\nu)$ as  in \eqref {gi4}  we have 
\begin{equation}
\left\vert A^{i}(x,\nu)-A^{i}(z,\nu)\right\vert \leq L\left\vert
x-z\right\vert \ ,\qquad x,z\in \mathcal{T}_{D},\ i=1,..,d,\qquad L=2K\ .
\end{equation}
\end{lem}

\begin{proof}
$\forall i=1,..,d,$ we have
\begin{equation}
\label{d.16}%
\begin{split}
(\nabla_{x}A^{i})(x,\nu)=  &  \frac{\int \nu(dy,du)\nu
(dy^{\prime},du^{\prime})u^{i}\left[  \nabla_{x}U(x-y)U(x-y^{\prime
})-U(x-y)\nabla_{x}U(x-y^{\prime})\right]  }{[(U\star\nu)(x)]^{2}}.
\end{split}
\end{equation}
Taking into account that $\left\vert u\right\vert \leq1$ we have
\begin{align}
\left\vert (\nabla_{x}A^{i})(x,\nu)\right\vert  &  \leq2\frac{\int \nu(dy,du)\nu(dy^{\prime},du^{\prime})\left\vert \nabla
_{x}U(x-y)\right\vert U(x-y^{\prime})}{[(U\star\nu)(x)]^{2}}\nonumber\\
&  =2\frac{\int \nu(dy,du)\nu(dy^{\prime},du^{\prime
})\left\vert \nabla_{x}U(x-y)\right\vert \frac{U(x-y)}{U(x-y)}U(x-y^{\prime}%
)}{[(U\star\nu)(x)]^{2}}\label{d.17}\\
&  \leq2\sup_{y}\frac{\left\vert \nabla_{x}U(x-y)\right\vert }{U(x-y)}%
\leq2K\ .\nonumber
\end{align}

\end{proof}

\vskip0.5cm

\begin{lem}
\label{le3} Let $U$ be the   interaction defined on $\mathcal{T}_{D}$ through
the periodization of $\tilde{U}$ as  defined in \eqref {ginevra2}.
For $\nu\in{\mathcal{M}}$ and $A(\cdot,\nu)$ as  in \eqref {gi4} we have
\begin{equation}
|A^{i}(x,\nu)-A^{i}(z,\nu)|\leq L\left\vert x-z\right\vert \ ,\qquad
x,z\in{\mathcal{T}}_{D},\quad i=1,..,d,\qquad L=\frac{D}{R^{2}}\ .
\end{equation}

\end{lem}

\begin{proof}
Let us write
\begin{align}
\left\vert A^{i}(x,\mu)-A^{i}(z,\mu)\right\vert  &  =\left\vert \int_{0}%
^{1}ds\frac{d}{ds}A^{i}(sx+(1-s)z,\mu)\right\vert \label{lip1}\\
&  \leq\sup_{s\in\lbrack0,1]}\left\vert \frac{d}{ds}A^{i}(sx+(1-s)z,\mu
)\right\vert \ ,\qquad i=1,..,d,\nonumber
\end{align}
and set $x_{0}=sx+(1-s)z$. We obtain
\[
\frac{d}{ds}A^{i}(x_{0},\mu) =\frac{(x-z)}{R}C_{i}(x_{0},R,\mu)
\]
where
\begin{equation}
C_{i}(x_{0},R,\mu)=\frac{\int\frac{(y^{\prime}-y)}{R}U(x_{0}-y^{\prime
})U(x_{0}-y)\int\left(  u^{\prime}\right)^{i}  \mu(dy^{\prime},du^{\prime}%
)\mu(dy,du)}{\int U(x_{0}-y^{\prime})U(x_{0}-y)\int\mu_{t}(dy^{\prime
},du^{\prime})\mu(dy,du)} . \label{lip5}%
\end{equation}
Recalling that $|u|\leq1,$ we obtain
\begin{equation}
|C(x_{0},R,\mu)|\leq\frac{D}{R}\ , \label{lip6}%
\end{equation}
since in the torus $\left\vert y^{\prime}-y\right\vert \leq D$ and the result
follows by (\ref{lip1}).
\end{proof}
 
For any $v\in B_{1},\nu\in{\mathcal{M}},M_{\epsilon}(\cdot,v,\nu)$ is easily
seen to be Lipschitz continuous in $X$. In fact we have the following:

\begin{lem}
\label{le5} Let $\nu\in{\mathcal{M}},\epsilon>0,U(\cdot)$ a smooth interaction
whose support contained in a ball of radius $R$ such that $\sup_{x\in B_{R}%
}\left\vert \nabla U(x)\right\vert \leq1$ and $M_{\epsilon}(\cdot,\cdot,\nu)$
as in (\ref{s1e}). Then, for any $v\in B_{1},M_{\epsilon}(\cdot,v,\nu)$ is
Lipschitz continuous in $X$:%
\begin{equation}
\left\vert M_{\epsilon}^{i}(x,v,\nu)-M_{\epsilon}^{i}(y,v,\nu)\right\vert \leq
L\left\vert x-y\right\vert \ ,\qquad x,y\in X,i=1,..,d,\qquad L=\frac
{2}{\epsilon}\ .
\end{equation}

\end{lem}

To prove the existence of the solution of (\ref{b1}) we prescribe a curve
$t\rightarrow\mu_{t}\in{\mathcal{M}}$ weakly continuous in $t$ and we consider
the following non-autonomous system of ordinary differential equations:
\begin{equation}
\left\{
\begin{array}
[c]{l}%
\frac{d}{dt}x(t)=v(t)\\
\frac{d}{dt}v(t)=M(x(t),v(t),\mu_{t})
\end{array}
\right.  \ . \label{od1}%
\end{equation}
Under the assumption that $M(\cdot,\cdot,\mu_{t})$ is Lipschitz continuous in
$X\times B_{1}$ there exists an unique global solution of (\ref{od1}) for any
given initial datum. The corresponding time dependent two parameters flow is
denoted by $T_{t,s}[\mu_{\cdot}].$ Under this time dependent flow any initial
measure evolves as
\begin{equation}
\nu_{t}=\nu_{0}\circ T_{0,t}[\mu_{\cdot}]\ , \label{p1}%
\end{equation}
where $\nu_{0}\circ T_{0,t}[\mu_{\cdot}]$ is the push forward of the measure
$\nu_{0}$ under the flow. For any test function $g$ we have that
\begin{equation}
\nu_{t}(g)=\nu_{0}(g\circ T_{t,0}[\mu_{\cdot}])\ , \label{p2}%
\end{equation}
where $g\circ T_{t,0}[\mu_{\cdot}]$ is the pull back under the flow of any test functions $g
$. By the existence and uniqueness of the solution of
(\ref{od1}) for any initial datum, the inverse flow $(T_{t,s}[\mu_{\cdot
}])^{-1}$ is well defined. The equation for the evolution of $\nu_{t},$ easily
derived, is
\begin{align}
\frac{\partial(\nu_{t}(g))}{\partial t}  &  =\frac{\partial(\nu_{0}(g\circ
T_{t,0}[\mu_{\cdot}]))}{\partial t}=\nu_{0}((v\nabla_{x}g)\circ T_{t,0}%
[\mu_{\cdot}])+\nu_{0}((M(x,v,\mu_{t})\nabla_{v}g)\circ T_{t,0}[\mu_{\cdot
}])\\
&  =\nu_{t}(v\cdot\nabla_{x}g)+\nu_{t}(M(x,v,\mu_{t})\cdot\nabla
_{v}g)\ .\nonumber
\end{align}
One immediately realizes that proving the existence and uniqueness of the
solution of (\ref{b1}) is equivalent to prove the existence of a fixed point
for the time dependent flow $\mu_{t}=\mu_{0}\circ T_{0,t}[\mu_{\cdot}].$ This
is the content of the next theorem.

\begin{thm}
\label{th1} Let $U$ be as in Lemma \ref{le2} or as in Lemma \ref{le3} and let
$M(\cdot,\cdot,\nu)$ be defined as in (\ref{s1}) for any $\nu\in{\mathcal{M}%
}({\mathcal{T}}_{D}\times B_{1})$. The equation (\ref{b1}) has an unique
solution in the space ${\mathcal{M}}({\mathcal{T}}_{D}\times B_{1}) $ if
$\mu_{0}\in{\mathcal{M}}({\mathcal{T}}_{D}\times B_{1}).$ Furthermore, take
two solutions of (\ref{b1}), $\mu_{t}$ starting at $\mu_{0}=\mu$ and $\nu_{t}$
starting at $\nu_{0}=\nu$ then in the bounded Lipschitz distance
\begin{equation}
d_{b{\mathcal{L}}}(\nu_{t},\mu_{t})\leq e^{ct}d_{b{\mathcal{L}}}(\mu,\nu)\ ,
\label{ine1}%
\end{equation}
where $c$ is a constant which depends on the Lipschitz constant of
$M(\cdot,\cdot,\nu)$ and on $\inf_{x\in{\mathcal{T}}_{D}}U(x)=:a>0.$

\end{thm}

The proof is obtained adapting the method explained in    \cite[Chapter 5]{S}
to our context.   To
facilitate the reader we report the proof of Theorem \ref{th1} in the Appendix.

\begin{rem}
\label{lun28} Theorem \ref{th1} does not hold in ${\mathbb{R}}^{d}\times
B_{1}$ when $U$ satisfies Lemma \ref{le2}. Although in this case $U$ is
globally Lipschitz continuous in ${\mathbb{R}}^{d},$ we are not able to show
that $M(x,v,\cdot)$ when $x\in{\mathbb{R}}^{d},v\in B_{1}$ is Lipschitz
continuous with respect to $\nu\in{\mathcal{M}}$ in the $d_{b{\mathcal{L}}}$
metric. The theorem applies with obvious modification to the $\epsilon
-$regularized system where $M$ is replaced by $M_{\epsilon}$ defined in
(\ref{s1e}) and holds either for the system defined on ${\mathcal{T}}%
_{D}\times B_{1}$ or on ${\mathbb{R}}^{d}\times B_{1}.$ The constant $c$ in
the statement of Theorem \ref{th1} will then depend on $\epsilon,$ the lower
bound of the denominator of $M_{\epsilon}.$
\end{rem}
The proof  of Theorem \ref  {cor1a} is an immediate consequence of   Theorem \ref {th1}.  The validity of Theorem  \ref{cor1a}
for the $\epsilon-$regularized system, where the \emph{local mean velocity
increment} is $M_{\epsilon},$ is immediate as well.

\begin{thm}
\label{th2} Let $M(\cdot,\cdot,\mu)$ be as in (\ref{s1}) and assume that
$M(\cdot,\cdot,\mu)\in C^{1}(X\times B_{1})$ for $\mu\in{\mathcal{M}}.$ If
$\mu_{0}(dx,dv)=f_{0}(x,v)dxdv,$ then $\mu_{t}(dx,dv)=f_{t}(x,v)dxdv$ and
$f_{t}$ is the weak solution of (\ref{eq1a}). Furthermore, if $f_{0}\in
C^{k}(X\times B_{1}),k\geq1,$ and $M(\cdot,\cdot,\mu)\in C^{k}(X\times B_{1})
$ for $\mu\in{\mathcal{M}},$ then $f_{t}\in C^{k}(X\times B_{1}).$
\end{thm}

\begin{proof}
We start showing that for any given weakly continuous curve $t\rightarrow
\mu_{t}\in{\mathcal{M}},$ if $\nu_{0}(dx,dv)=q_{0}(x,v)dxdv,$ i.e. absolutely
continuous with respect to the Lebesgue measure, then $\nu_{t}(dx,dv)=q_{t}%
(x,v)dxdv,$ where
\begin{equation}
\frac{\partial}{\partial t}q_{t}(x,v)+\nabla_{x}q_{t}(x,v)\cdot v+\nabla
_{v}\cdot\left[  M(x,v,\mu_{t})q_{t}(x,v)\right]  =0\ , \label{ep4}%
\end{equation}
and, if $q_{0}\in C^{k}(X\times B_{1})$ and $M(\cdot,v,\mu)\in C^{k}(X\times
B_{1})$ for any $\mu\in{\mathcal{M}},$ then $q_{t}\in C^{k}(X\times B_{1}).$
Note that (\ref{ep4}) corresponds to a linearization of (\ref{eq1a}) since
$M(x,v,\mu_{t})$ does not depend on $q_{\cdot}$ once $\mu_{t}$ is given. In
Theorem \ref{th1} we proved that the fixed point equation $\mu_{t}=\mu
_{0}\circ T_{0,t}[\mu_{\cdot}]$ holds. Therefore, by this result and the
validity of (\ref{ep4}), one immediately obtains that $\mu_{t}$ has density
and the thesis of the theorem is proven. We are then left with the proof of
(\ref{ep4}). Let us set $w=(x,v)\in X\times B_{1}.$ For any test function $g$
we obtain
\begin{align}
\nu_{t}(g)  &  =\nu_{0}\circ T_{0,t}[\mu_{\cdot}](g)=\nu_{0}(g\circ
T_{t,0}[\mu_{\cdot}])\label{ja1}\\
&  =\int_{X\times B_{1}}\nu_{0}(dw)(g\circ T_{t,0}[\mu_{\cdot}])(w)=\int
_{X\times B_{1}}q_{0}(w)\left(  g\circ T_{t,0}[\mu_{\cdot}]\right)
(w)dw\nonumber\\
&  =\int_{X\times B_{1}}q_{0}(w)\circ(T_{t,0}[\mu_{\cdot}])^{-1}{\mathcal{J}%
}(w,\mu_{t})g(w)dw\nonumber
\end{align}
where ${\mathcal{J}}(w,\mu_{t})=Det\left[  \partial_{\cdot}(T_{t,0}%
)[\mu_{\cdot}])^{-1}(w)\right]  $ is the Jacobian of the flow $(T_{t,0}%
[\mu_{\cdot}])^{-1}$ computed in $w.$ Since the divergence of the vector field
$(v(s),M(x,v,\mu_{s}))$ is given by
\begin{equation}
\sum_{i=1}^{d}\left[  \frac{\partial v^{i}}{\partial x^{i}}+\frac
{M^{i}(x,v,\mu_{s})}{\partial v^{i}}\right]  =-d\ , \label{orli1}%
\end{equation}
by Liouville Theorem (see \cite{Arn1} or \cite{Arn2}) for any weakly
continuous curve $t\rightarrow\mu_{t}\in{\mathcal{M}},$ we have
\begin{equation}
Det\left[  \partial_{\cdot}(T_{t,0})[\mu_{\cdot}])\right]  (w)=e^{-dt}%
\ ,\qquad\forall w\in X\times B_{1}\ ,
\end{equation}
hence,
\begin{equation}
{\mathcal{J}}(w,\mu_{t})=e^{dt}\ ,\qquad\forall w\in X\times B_{1}\ .
\label{ja3}%
\end{equation}
Then, from (\ref{ja1}) and (\ref{ja3}), we obtain
\begin{equation}
\nu_{t}(g)=\int_{X\times B_{1}}e^{dt}\left(  q_{0}\circ(T_{t,0}[\mu_{\cdot
}])^{-1}\right)  (w)g(w)dw=\int_{X\times B_{1}}q_{t}(w)g(w)dw\ , \label{ja2}%
\end{equation}
where we denote by
\begin{equation}
q_{t}(w):=e^{dt}\left(  q_{0}\circ(T_{t,0}[\mu_{\cdot}])^{-1}\right)  (w)\ .
\label{ep1}%
\end{equation}
Notice that $q_{t}(w)$ is weakly continuous in time, since $\mu_{\cdot}$ is
weakly continuous. Furthermore, if, for $k\geq1,M(\cdot,\cdot,\mu)\in
C^{k}(X\times B_{1}),\mu\in{\mathcal{M}}$ and $q_{0}\in C^{k}(X\times B_{1}),
$ then $q_{t}\in C^{k}(X\times B_{1}).$ Writing $e^{-dt}(q_{t}\circ
T_{t,0}[\mu_{\cdot}])(w)=q_{0}(w)$ and differentiating with respect to $t$ we
get
\begin{gather}
\frac{\partial}{\partial t}\left(  e^{-dt}q_{t}\circ(T_{t,0}[\mu_{\cdot
}])(w)\right)  =-de^{-dt}(q_{t}\circ T_{t,0}[\mu_{\cdot}])(w)+\\
+e^{-dt}\frac{\partial}{\partial t}(q_{t}\circ T_{t,0}[\mu_{\cdot
}])(w)+e^{-dt}\nabla_{x}(q_{t}\circ T_{t,0}[\mu_{\cdot}])(w)\cdot v\circ
T_{t,0}[\mu_{\cdot}]+\nonumber\\
+e^{-dt}\nabla_{v}\left(  (q_{t}\circ T_{t,0}[\mu_{\cdot}])(w)\right)
\cdot(M(\cdot,\cdot,\mu_{t})\circ T_{t,0}[\mu_{\cdot}])(w)=0\ .\nonumber
\end{gather}
Multiplying both members of the previous identity for $e^{dt}$ and applying to
$(T_{t,0}[\mu_{\cdot}])^{-1}w$ we obtain
\begin{equation}
-dq_{t}(x,v)+\frac{\partial}{\partial t}q_{t}(x,v)+\nabla_{x}q_{t}(x,v)\cdot
v+\nabla_{v}q_{t}(x,v)\cdot M(x,v,\mu_{t})=0\ .
\end{equation}
Notice that this last equation is linear in $q_{\cdot}$ since $\mu_{\cdot} $
is given. Therefore, the equation for $f_{t}$ is
\begin{equation}
\frac{\partial}{\partial t}f_{t}(x,v)-df_{t}(x,v)+\nabla_{x}f_{t}(x,v)\cdot
v+M(x,v,f_{t})\cdot\nabla_{v}f_{t}(x,v)=0
\end{equation}
which corresponds to (\ref{eq1a}).
\end{proof}

\begin{rem}
\label{no2} Theorem \ref{th2} holds also for the $\epsilon-$regularized system.
In such a case  the\emph{\ local mean velocity increment} is $M_{\epsilon},$ see
(\ref{s1e}),  and  the phase space of the system can  be either
${\mathbb{R}}^{d}\times B_{1}$ or ${\mathcal{T}}_{D}\times B_{1}.$ The  
difference  from  the computations  in  Theorem \ref{th2} is that
\begin{equation}
\nabla\cdot M_{\epsilon}(x,v,\mu_{s})=-d\frac{(U\star\mu_{s})(x)}{(U\star
\mu_{s})(x)+\epsilon}:=-dh_{s}^{\epsilon}(x)\equiv-dh^{\epsilon}(\mu
_{s})(x)\ . \label{orli3}%
\end{equation}
For any weakly continuous curve $t\rightarrow\mu_{t}\in{\mathcal{M}}$
and any $w\in X\times B_{1},$ by Liouville Theorem, we have
\begin{equation}
\operatorname{Det}\left[  \partial_{\cdot}(T_{t,0}[\mu_{\cdot}])\right]
\left(  w\right)  =e^{-d\int_{0}^{t}ds(h_{s}^{\epsilon}\circ T_{s,0}%
[\mu_{\cdot}])(w)}\ ,
\end{equation}
therefore
\begin{equation}
{\mathcal{J}}(w,\mu_{t})=\operatorname{Det}\left[  \partial_{\cdot}%
[(T_{t,0}[\mu_{\cdot}])^{-1}]\right]  =e^{d\int_{0}^{t}ds(h_{s}^{\epsilon
}\circ T_{s,t}[\mu_{\cdot}])(w)}\ . \label{ja3a}%
\end{equation}
Then, from (\ref{ja1}) and (\ref{ja3a}), we obtain
\begin{align}
\nu_{t}(g)  &  =\int_{X\times B_{1}}e^{d\int_{0}^{t}ds(h_{s}^{\epsilon}\circ
T_{s,t}[\mu_{\cdot}])(w)}\left(  q_{0}\circ(T_{t,0}[\mu_{\cdot}])^{-1}\right)
(w)g(w)dw\label{ja2a}\\
&  =\int_{X\times B_{1}}q_{t}(w)g(w)dw\ ,\nonumber
\end{align}
where  
\begin{equation}
q_{t}(w):=e^{d\int_{0}^{t}ds(h_{s}^{\epsilon}\circ T_{s,t}[\mu_{\cdot}%
])(w)}\left(  q_{0}\circ(T_{t,0}[\mu_{\cdot}])^{-1}\right)  (w)\ .
\label{ep1a}%
\end{equation}
Notice that $q_{t}(w)$ is weakly continuous in time, since $\mu_{\cdot}$ is
weakly continuous. Furthermore, if $M_{\epsilon}(\cdot,\cdot,\mu)\in
C^{k}(X\times B_{1}),\mu\in{\mathcal{M}}$ and $q_{0}\in C^{k}(X\times B_{1}),
$ then $q_{t}\in C^{k}(X\times B_{1}).$ Writing
\begin{equation}
e^{-d\int_{0}^{t}ds(h_{s}^{\epsilon}\circ T_{s,0}[\mu_{\cdot}])(w)}(q_{t}\circ
T_{t,0}[\mu_{\cdot}])(w)=q_{0}(w)
\end{equation}
and differentiating with respect to $t$ we get
\begin{multline}
\frac{\partial}{\partial t}\left(  e^{-d\int_{0}^{t}ds(h_{s}^{\epsilon}\circ
T_{s,0}[\mu_{\cdot}])(w)}(q_{t}\circ(T_{t,0}[\mu_{\cdot}])(w)\right)
=-d(h_{t}^{\epsilon}\circ T_{t,0}[\mu_{\cdot}])(w)e^{-d\int_{0}^{t}%
ds(h_{s}^{\epsilon}\circ T_{s,0}[\mu_{\cdot}])(w)}(q_{t}\circ T_{t,0}%
[\mu_{\cdot}])(w)\nonumber\\
+e^{-d\int_{0}^{t}ds(h_{s}^{\epsilon}\circ(T_{s,0}[\mu_{\cdot}])(w)}%
\frac{\partial}{\partial t}(q_{t}\circ T_{t,0}[\mu_{\cdot}])(w)+e^{-d\int
_{0}^{t}ds(h_{s}^{\epsilon}\circ T_{s,0}[\mu_{\cdot}])(w)}\nabla_{x}%
q_{t}(w)\circ T_{t,0}[\mu_{\cdot}]\cdot v\circ T_{t,0}[\mu_{\cdot}]\nonumber\\
+e^{-d\int_{0}^{t}ds(h_{s}^{\epsilon}(\circ T_{s,0}[\mu_{\cdot}])(w)}%
\nabla_{v}\left(  (q_{t}\circ T_{t,0}[\mu_{\cdot}])(w)\right)  \cdot
(M_{\epsilon}(\cdot,\cdot,\mu_{t})\circ T_{t,0}[\mu_{\cdot}])(w)=0\ .
\end{multline}
Multiplying by $e^{d\int_{0}^{t}ds(h_{s}^{\epsilon}\circ(T_{s,0}[\mu_{\cdot
}])(w)}$ and applying to $(T_{t,0}[\mu_{\cdot}])^{-1}w$ we obtain
\begin{equation}
-dh(\mu_{t})(x)q_{t}(x,v)+\frac{\partial}{\partial t}q_{t}(x,v)+\nabla
_{x}q_{t}(x,v)\cdot v+\nabla_{v}q_{t}(x,v)\cdot M_{\epsilon}(x,v,\mu_{t})=0,
\end{equation}
which is linear in $q_{\cdot}$ since $\mu_{\cdot}$ is given. Therefore the
equation for $f_{t}$ is
\begin{equation}
\frac{\partial}{\partial t}f_{t}(x,v)-dh^{\epsilon}(f_{t})(x)f_{t}%
(x,v)+\nabla_{x}f_{t}(x,v)\cdot v+M_{\epsilon}(x,v,f_{t})\cdot\nabla_{v}%
f_{t}(x,v)=0\ ,
\end{equation}
which corresponds to (\ref{eq1a}) with $M$ replaced by $M_{\epsilon},$ taking
into account the definition of $h^{\epsilon}$ in (\ref{orli3}).
\end{rem}

\subsection{Qualitative behaviour of the solution of (\ref{b1})\label{Sb1}}

\begin{lem}
\label{le1} Let $M$ as in (\ref{s1}) and $t\rightarrow\mu_{t}\in{\mathcal{M}}$
be the solution of (\ref{b1}) with initial datum $\mu_{0}.$
We have%

\begin{equation}
\mu_{t}(x)=\mu_{0}(x)+\int_{0}^{t}\mu_{s}(v)ds\ , \label{ma.1}%
\end{equation}

\begin{equation}
\int_{X\times B_{1}}\left\vert v\right\vert ^{2}\mu_{t}(dx,dv)\leq
\int_{X\times B_{1}}\left\vert v\right\vert ^{2}\mu_{0}(dx,dv)\ . \label{g.1}%
\end{equation}

\end{lem}

\begin{proof}
By (\ref{b1}) we have
\begin{equation}
\frac{d}{dt}\int_{X\times B_{1}}x^{i}\mu_{t}(dx,dv)=\mu_{t}(v^{i})\ ,\qquad
i=1,..,d\ ,
\end{equation}
which implies (\ref{ma.1}). To obtain (\ref{g.1}), again from equation
(\ref{b1}) we get
\begin{align}
\frac{d}{dt}\int_{X\times B_{1}}\left\vert v\right\vert ^{2}\mu_{t}(dx,dv)  &
=\frac{d}{dt}\mu_{t}(\left\vert v\right\vert ^{2})\label{d.14}\\
&  =\mu_{t}(v\cdot\nabla_{x}\left\vert v\right\vert ^{2})+\mu_{t}%
(M(\cdot,\cdot,\mu_{t})\cdot\nabla_{v}\left\vert v\right\vert ^{2})\nonumber\\
&  =2\sum_{i=1}^{d}\mu_{t}(M^{i}(\cdot,\cdot,\mu_{t})v^{i})\leq0\ .\nonumber
\end{align}
Namely, for $i=1,..,d,$ when $M^{i}(\cdot,\cdot,\mu_{t})\neq0,$ we obtain
\begin{align}
&  \mu_{t}(M^{i}(\cdot,\cdot,\mu_{t})v^{i})=\int_{X\times B_{1}}\mu
_{t}(dx,dv)M^{i}(x,v,\mu_{t})v_{i}\label{d.15}\\
&  =\int_{X\times B_{1}}\mu_{t}(dx,dv)\left(  \frac{\int_{X\times B_{1}%
}U(x-y)\left(  v^{i}u^{i}-\left(  v^{i}\right)  ^{2}\right)  \mu_{t}%
(dy,du)}{\int_{X\times B_{1}}U(x-y)\mu_{t}(dy,du)}\right) \nonumber\\
&  =\int_{(X\times B_{1})^{2}}\mu_{t}(dx,dv)\mu_{t}(dy,du)\frac{U(x-y)v^{i}%
u^{i}}{\int_{X\times B_{1}}U(x-y)\mu_{t}(dy,du)}-\int_{X\times B_{1}}\mu
_{t}(dx,dv)\left(  v^{i}\right)  ^{2}\leq0\ ,\nonumber
\end{align}
by Schwartz inequality.
\end{proof}
The Jensen inequality and (\ref{g.1}) imply the boundedness of the mean velocity
$\mu_{t}(v).$
   
\nada{
\begin{lem}
\label{orli5} Take $M$ as in (\ref{s1}). The equation (\ref{eq1a}) is not time
reversible, i.e. invariant under simultaneous reflection $t\rightarrow-t $ and
$v\rightarrow-v.$
\end{lem}

\begin{proof}
Let us set $b_{t}(x,v):=f_{-t}(x,-v).$ We have that $\frac{\partial}{\partial
t}b_{t}(x,v)=-\frac{\partial}{\partial t}f_{-t}(x,-v)$ and $\frac
{\partial{b_{t}(x,v)}}{\partial v^{i}}=-\frac{\partial{f_{-t}(x,-v)}}{\partial
v^{i}}.$ Therefore, by (\ref{eq1a}),
\begin{equation}
\frac{\partial}{\partial t}b_{t}(x,v)=-\frac{\partial}{\partial t}%
f_{-t}(x,-v)=-v\cdot\nabla_{x}f_{-t}(x,-v)-df_{-t}(x,-v)+\sum_{i=1}^{d}%
M^{i}(x,-v,f_{-t})\frac{\partial f_{-t}(x,-v)}{\partial v^{i}}\
\end{equation}
and
\begin{align}
M(x,-v,f_{-t})  &  =\left(  \frac{\int_{X\times{\mathbb{R}}^{d}}U(x-y)\left(
u+v\right)  f_{-t}(y,u)dydu}{\int_{X\times B_{1}}U(x-y)f_{-t}(y,u)dydu}\right)
\\
&  =\left(  \frac{\int_{X\times B_{1}}U(x-y)\left(  -u+v\right)
b_{t}(y,u)dydu}{\int_{X\times B_{1}}U(x-y)b_{t}(y,u)dydu}\right)
=-M(x,v,b_{t})\ .\nonumber
\end{align}
Thus, the equation for $b_{t}(x,v)$ is
\[
\frac{\partial}{\partial t}b_{t}(x,v)+v\cdot\nabla_{x}b_{t}(x,v)+db_{t}%
(x,v)-\sum_{i=1}^{d}M^{i}(x,v,b_{t})\frac{\partial b_{t}(x,v)}{\partial v^{i}%
}=0
\]
which differs from equation (\ref{eq1a}).
\end{proof}

\begin{rem}
Lemma \ref{orli5} holds also for the $\epsilon-$regularized system. In such a case    $M$
is replaced by $M_{\epsilon}.$
\end{rem}
}

\vskip1cm

Let $f_{t}$ be the solution at time $t$ of the equation (\ref{eq1a}). We denote
by $H(f_{t})$ the Boltzmann-Vlasov entropy
\begin{equation}
H(f_{t}):=-\int_{X\times B_{1}}f_{t}(x,v)\ln(f_{t}(x,v))dxdv\ . \label{ent1}%
\end{equation}
 In the next lemma we show that  the  Boltzmann-Vlasov entropy
 $H(f_{t})$ is a decreasing function of time.
 Notice that   equation (\ref{eq1a}) is not time
reversible, i.e. invariant under simultaneous reflection $t\rightarrow-t $ and
$v\rightarrow-v.$ 
\begin{lem}
\label{entro1} Let $f_{\cdot}$ be the solution of (\ref{eq1a}) with $M$ chosen
as in (\ref{s1}), then
\begin{equation}
\frac{d}{dt}H(f_{t})=-d\ . \label{orli6}%
\end{equation}
Let $f_{\cdot}^{\epsilon}$ be the solution of (\ref{eq1a}) with $M$ replaced by
$M_{\epsilon}$ chosen as in (\ref{s1e}), then
\begin{equation}
\frac{d}{dt}H(f_{t}^{\epsilon})=-d\int_{X\times B_{1}}h_{t}^{\epsilon}%
(x)f_{t}^{\epsilon}(x,v)dxdv\ , \label{orli7}%
\end{equation}
where, as in (\ref{orli3}), $h_{t}^{\epsilon}=\frac{(U\star f_{t}^{\epsilon}%
)}{(U\star f_{t}^{\epsilon})+\epsilon}.$
\end{lem}

\begin{proof}
We start showing (\ref{orli6}). The proof of (\ref{orli7}) is similar and we
will only outline the differences.
\begin{align}
\frac{d}{dt}H(f_{t})  &  =-\int_{X\times B_{1}}\frac{\partial f_{t}}{\partial
t}(x,v)\left[  \left(  \ln f_{t}(x,v)\right)  +1\right]  dxdv\label{ent2}\\
&  =\int_{X\times B_{1}}\left(  \ln f_{t}(x,v)\right)  \left[  v\cdot
\nabla_{x}f_{t}(x,v)+\nabla_{v}\cdot M(x,v,t)f_{t}(x,v)\right]
dxdv\ .\nonumber
\end{align}
Integrating by part the last term in (\ref{ent2}) we get
\begin{align}
\frac{d}{dt}H(f_{t})  &  =-\int_{X\times B_{1}}\nabla_{x}\left(  \ln
f_{t}(x,v)\right)  \cdot vf_{t}(x,v)dxdv\label{ent3}\\
&  -\int_{X\times B_{1}}\nabla_{v}\left(  \ln f_{t}(x,v)\right)  \cdot\left[
M(x,v,f_{t})f_{t}(x,v)\right]  dxdv\nonumber\\
&  =-\int_{X\times B_{1}}\nabla_{x}f_{t}(x,v)\cdot vdxdv-\int_{X\times B_{1}%
}\nabla_{v}f_{t}(x,v)\cdot\left[  M(x,v,f_{t})\right]  dxdv\ .\nonumber
\end{align}
The first integral gives zero contribution since $\int_{X\times B_{1}}%
f_{t}(x,v)dxdv=1$ for all $t>0,$ i.e. $f_{t}\in L^{1}(X\times B_{1}).$ For the
second term notice that $\nabla_{v}\cdot\left[  M(x,v,f_{t})\right]  =-d,$
therefore
\begin{align}
\int_{X\times B_{1}}\nabla_{v}f_{t}(x,v)\cdot\left[  M(x,v,f_{t})\right]
dxdv  &  =-\int_{X\times B_{1}}f_{t}(x,v)\nabla_{v}\cdot\left[  M(x,v,f_{t}%
)\right]  dxdv\label{ent4}\\
&  =d\int_{X\times B_{1}}f_{t}(x,v)dxdv=d\ .\nonumber
\end{align}
We then obtain (\ref{orli6}). To get (\ref{orli7}) we proceed in the same way.
We need only to modify (\ref{ent4}) as
\begin{align}
\int_{X\times B_{1}}\nabla_{v}f_{t}^{\epsilon}(x,v)\cdot\left[  M_{\epsilon
}(x,v,f_{t}^{\epsilon})\right]  dxdv  &  =-\int_{X\times B_{1}}f_{t}%
^{\epsilon}(x,v)\nabla_{v}\cdot\left[  M_{\epsilon}(x,v,f_{t}^{\epsilon
})\right]  dxdv\label{ent4a}\\
&  =d\int_{X\times B_{1}}h_{t}^{\epsilon}(x)f_{t}^{\epsilon}%
(x,v)dxdv\ .\nonumber
\end{align}

\end{proof}

By the above lemma,
\begin{equation}
\lim_{t\rightarrow\infty}H(f_{t})=-\infty\ .
\end{equation}
From this we can deduce that even starting at time $t=0$ from a measure which
is absolutely continuous with respect to Lebesgue measure in $X\times B_{1},$
having therefore finite Boltzmann-Vlasov entropy, at infinity the asymptotic
measure is singular with respect to the Lebesgue one.
The same conclusions can be also drawn for the $\epsilon-$regularized system.

\section{Appendix}

\subsection{Proof of Lemma \ref{gi1}:}

Denote by $b_{i,j} (\cdot)$, for $i=1,\dots, N$ and $j=1, \dots, N$ the
elements of the matrix $B (\cdot)= \left\{  b_{i,j}(\cdot) \otimes \mathbb{I}%
_{d}\right\}  _{i,j=1,..,N} $ defined in \eqref  {roma2}. By definition of
$B(\mathbf{q}(t))$
\begin{align}
b_{i,j}(\mathbf{q}(t))  &  :=a_{i,j}(\mathbf{q}(t))-a_{i,j}(\mathbf{q}^{0}),
\end{align}
where $\{a_{i,j}(\cdot) \}$ are defined in \eqref {v1}.
Writing $a_{i,j}(\mathbf{q}(t))$ as
\begin{equation}
a_{i,j}(\mathbf{q}(t))=a_{i,j}(\mathbf{q}^{0})+\int_{0}^{1}ds\frac{d}%
{ds}a_{i,j}((1-s)\mathbf{q}^{0}+s\mathbf{q}(t)),
\end{equation}
we have
\begin{equation}
b_{i,j}(\mathbf{q}(t))=\int_{0}^{1}ds\frac{d}{ds}a_{i,j}((1-s)\mathbf{q}%
_{N}^{0}+s\mathbf{q}_{N}(t))\ .
\end{equation}
Therefore, setting
\begin{equation}
x_{i,j}(s,t):=(1-s)(q_{i}^{0}-q_{j}^{0})+s\left(  q_{i}(t)-q_{j}(t)\right)
\ ,\quad i,j=1,..,N\ ,
\end{equation}
we have%
\begin{align}
\frac{d}{ds}a_{i,j}((1-s)\mathbf{q}^{0}+s\mathbf{q}(t))  &  =\frac{d}%
{ds}\left(  \frac{U ((1-s)(q_{i}^{0}-q_{j}^{0})+s(q_{i}(t)-q_{j}(t)))}%
{\sum_{k=1}^{N}U ((1-s)(q_{i}^{0}-q_{k}^{0})+s(q_{i}(t)-q_{k}(t)))}\right)
\label{de.1}\\
&  =\frac{d}{ds}\left(  \frac{U (x_{i,j}(s,t))}{\sum_{k=1}^{N}U (x_{i,k}%
(s,t))}\right) \nonumber\\
&  =\frac{\nabla U (x_{i,j}(s,t))\cdot\left[  -(q_{i}^{0}-q_{j}^{0}%
)+q_{i}(t)-q_{j}(t)\right]  }{\sum_{k=1}^{N}U (x_{i,k}(s,t))}\nonumber\\
&  -\frac{U (x_{i,j}(s,t))\sum_{k=1}^{N}\nabla U (x_{i,k}(s,t))\cdot\left[
-(q_{i}^{0}-q_{k}^{0})+q_{i}(t)-q_{k}(t)\right]  }{\left(  \sum_{k=1}^{N}U
(x_{i,k}(s,t))\right)  ^{2}}\ .\nonumber
\end{align}
Hence,
\begin{align}
\sum_{j=1}^{N}|b_{i,j}(\mathbf{q}(t))|  &  \leq2\frac{\sum_{j=1}^{N}\left\vert
\nabla U(x_{i,j}(s,t))\right\vert \left\vert -(q_{i}^{0}-q_{j}^{0}%
)+q_{i}(t)-q_{j}(t)\right\vert }{\sum_{k=1}^{N}U (x_{i,k}(s,t))}\label{de.2}\\
&  \leq2\frac{N}{U\left(  0\right)  +( N-1) \eta(\mathbf{q}(t), \mathbf{q}%
^{0}) }\sup_{x\in\mathbb{R}^{d}}\left\vert \nabla U (x)\right\vert
\max_{i,j\in\{1,..,N\}}\left\vert -(q_{i}^{0}-q_{j}^{0})+q_{i}(t)-q_{j}%
(t)\right\vert \ ,\nonumber
\end{align}
where
\begin{equation}
\label{roma1}\eta(\mathbf{q}(t), \mathbf{q}^{0}) = \inf_{s \in[0,1]} \inf_{i,
k \in\{1, \dots, N\}} U((1-s)(q_{i}^{0}-q_{k}^{0})+s\left(  q_{i}%
(t)-q_{k}(t)\right)  ) \ge0.
\end{equation}
Since
\begin{equation}
\left\Vert B(\mathbf{q}(t))\right\Vert \leq\left\Vert \tilde{B}\left(
\mathbf{q}\left(  t\right)  \right)  \right\Vert _{\infty}=\max_{i=1,..,N}%
\sum_{j=1}^{N}\left\vert b_{i,j}(\mathbf{q}(t))\right\vert
\end{equation}
we get the thesis.

\subsection{Proof of Theorem \ref{th1}:}

We  adapt to our model  \cite[Theorem 5.1]{S} and divide the proof in two steps.

\textit{Step 1} We start proving (\ref{ine1}). Assume that $\nu_{t}$ and
$\mu_{t}$ solve (\ref{mea1}). We have, by the triangular inequality, that
\begin{align}
d_{b{\mathcal{L}}}(\nu_{t},\mu_{t})  &  =d_{b{\mathcal{L}}}(\nu_{0}\circ
T_{0,t}[\nu_{\cdot}],\mu_{0}\circ T_{0,t}[\mu_{\cdot}])\label{mea2}\\
&  \leq d_{b{\mathcal{L}}}(\mu_{0}\circ T_{0,t}[\nu_{\cdot}],\mu_{0}\circ
T_{0,t}[\mu_{\cdot}])+d_{b{\mathcal{L}}}(\mu_{0}\circ T_{0,t}[\nu_{\cdot}%
],\nu_{0}\circ T_{0,t}[\nu_{\cdot}])\ .\nonumber
\end{align}
Denote by $w:=(x,v),V(\mu_{\cdot})_{s}(w):=(v(s),A(x(s),\mu_{s})-v(s))$ the
vector field on the right hand side of (\ref{od1}). The second term can be
bounded as
\begin{align}
d_{b{\mathcal{L}}}(\mu_{0}\circ T_{0,t}[\nu_{\cdot}],\nu_{0}\circ T_{0,t}%
[\nu_{\cdot}])  &  =e^{Lt}\sup_{f\in{\mathcal{D}}}\left\vert \int
_{\mathcal{T}_{D}\times B_{1}}[d\mu_{0}-d\nu_{0}]\left(  e^{-Lt}f\circ
T_{t,0}[\nu_{\cdot}]\right)  \right\vert \label{mea3}\\
&  \leq e^{Lt}d_{b{\mathcal{L}}}(\mu_{0},\nu_{0})\nonumber
\end{align}
where $L$ is the Lipschitz constant of $V(\mu_{\cdot})_{s}(\cdot).$ Notice
that the Lipschitz bound of $V(\mu_{\cdot})_{s}(\cdot)$ can be easily derived
from the Lipschitz bound of $A(\cdot,\mu_{\cdot}).$ We get (\ref{mea3}) if we
can show, since $f\in{\mathcal{D}},$ that $e^{-Lt}f\circ T_{t,0}[\nu_{\cdot}]$
is Lipschitz continuous with constant one and therefore it belongs to
${\mathcal{D}}.$ Let $w(t)=(x(t),v(t))$ be the solution of (\ref{od1}) with
initial condition $w_{0}=(x_{0},v_{0})$ and let $\tilde{w}(t)$ be the solution
of (\ref{od1}) with initial condition $\tilde{w}_{0}=(\tilde{x}_{0},\tilde
{v}_{0}),$ then we need to show that
\begin{equation}
|f(w(t))-f(\tilde{w}(t))|\leq C(t)|w_{0}-\tilde{w}_{0}|\ , \label{mea4}%
\end{equation}
with $C(t)\leq e^{Lt}.$ Writing
\begin{equation}
w(t)=w_{0}+\int_{0}^{t}V(\mu_{\cdot})_{s}(w(s))
\end{equation}
and
\begin{equation}
\tilde{w}(t)=\tilde{w}_{0}+\int_{0}^{t}V(\mu_{\cdot})_{s}(\tilde{w}(s))\ ,
\end{equation}
since $f\in{\mathcal{D}},$ we have
\begin{equation}
\ |f(w(t))-f(\tilde{w}(t))|\leq|w(t)-\tilde{w}(t)|\ . \label{mea6}%
\end{equation}
Furthermore,
\begin{align}
|w(t)-\tilde{w}(t)|  &  \leq|w_{0}-\tilde{w}_{0}|+\int_{0}^{t}|V(\mu_{\cdot
})_{s}(w(s))-V(\mu_{\cdot})_{s}(\tilde{w}(s))|ds\label{mea5}\\
&  \leq|w_{0}-\tilde{w}_{0}|+L\int_{0}^{t}|w(s)-\tilde{w}(s)|ds\ .\nonumber
\end{align}
By the Gronwall's inequality
\begin{equation}
|w(t)-\tilde{w}(t)|\leq e^{Lt}|w_{0}-\tilde{w}_{0}|
\end{equation}
proving $e^{-Lt}f\circ T_{t,0}[\nu_{\cdot}]\in{\mathcal{D}}$ and so
(\ref{mea3}). We are then left with the estimate the other term in
(\ref{mea2}) which, since $f\in{\mathcal{D}},$%
\begin{align}
d_{b{\mathcal{L}}}(\mu_{0}\circ T_{0,t}[\nu_{\cdot}],\mu_{0}\circ T_{0,t}%
[\mu_{\cdot}])  &  =\sup_{f\in{\mathcal{D}}}\left\vert \int_{\mathcal{T}%
_{D}\times B_{1}}d\mu_{0}\left\{  f\circ T_{t,0}[\nu_{\cdot}]-f\circ
T_{t,0}[\mu_{\cdot}]\right\}  \right\vert \label{mea8}\\
&  \leq\int_{\mathcal{T}_{D}\times B_{1}}\mu_{0}\left(  dw\right)  \left\vert
T_{t,0}[\nu_{\cdot}]w-T_{t,0}[\mu_{\cdot}]w\right\vert =:\lambda(t)\nonumber
\end{align}
where $T_{t,0}[\nu_{\cdot}]$ and $T_{t,0}[\mu_{\cdot}]$ are both solutions of
the equation (\ref{od1}) with the same initial conditions but with different
vector fields. We have
\begin{align}
\lambda(t)  &  =\int_{\mathcal{T}_{D}\times B_{1}}\mu_{0}(dw)\left\vert
\left\{  T_{t,0}[\nu_{\cdot}]w-T_{t,0}[\mu_{\cdot}]w\right\vert \right\}
\label{mea9}\\
&  =\int_{\mathcal{T}_{D}\times B_{1}}\mu_{0}(dw)\left\vert \int_{0}%
^{t}dsV(\nu_{\cdot})_{s}(T_{s,0}[\nu_{\cdot}]w)-\int_{0}^{t}dsV(\mu_{\cdot
})_{s}(T_{s,0}[\mu_{\cdot}]w)\right\vert \nonumber\\
&  \leq\int_{\mathcal{T}_{D}\times B_{1}}\mu_{0}(dw)\left\vert \int_{0}%
^{t}ds\left\{  V(\nu_{\cdot})_{s}(T_{s,0}[\nu_{\cdot}]w)-V(\nu_{\cdot}%
)_{s}(T_{s,0}[\mu_{\cdot}]w)\right\}  \right\vert \nonumber\\
&  +\int_{\mathcal{T}_{D}\times B_{1}}\mu_{0}(dw)\left\vert \int_{0}%
^{t}ds\left\{  V(\mu_{\cdot})_{s}(T_{s,0}[\mu_{\cdot}]w)-V(\nu_{\cdot}%
)_{s}(T_{s,0}[\mu_{\cdot}]w)\right\}  \right\vert \ .\nonumber
\end{align}
The first term of (\ref{mea9}) can be estimated by the Lipschitz property of
the vector field
\begin{align}
&  \int_{\mathcal{T}_{D}\times B_{1}}\mu_{0}(dw)\left\vert \int_{0}%
^{t}ds\left\{  V(\nu_{\cdot})_{s}(T_{s,0}[\nu_{\cdot}]w)-V(\nu_{\cdot}%
)_{s}(T_{s,0}[\mu_{\cdot}]w)\right\}  \right\vert \label{mea10}\\
&  \leq L\int_{\mathcal{T}_{D}\times B_{1}}\mu_{0}(dw)\int_{0}^{t}ds\left\vert
T_{s,0}[\nu_{\cdot}]w-T_{s,0}[\mu_{\cdot}]w]\right\vert \nonumber\\
&  =L\int_{0}^{t}ds\int_{\mathcal{T}_{D}\times B_{1}}\mu_{0}(dw)\left\vert
T_{s,0}[\nu_{\cdot}]w-T_{s,0}[\mu_{\cdot}]w]\right\vert =L\int_{0}^{t}%
\lambda(s)ds\ .\nonumber
\end{align}
For the second term of (\ref{mea9}) we have
\begin{align}
&  \int_{\mathcal{T}_{D}\times B_{1}}\mu_{0}(dw)\left\vert \int_{0}%
^{t}ds\left\{  V(\mu_{\cdot})_{s}(T_{s,0}[\mu_{\cdot}]w)-V(\nu_{\cdot}%
)_{s}(T_{s,0}[\mu_{\cdot}]w)\right\}  \right\vert \label{mea11}\\
&  \leq\int_{\mathcal{T}_{D}\times B_{1}}\mu_{0}(dw)\int_{0}^{t}ds\left\vert
V(\mu_{\cdot})_{s}\left(  T_{s,0}[\mu_{\cdot}]w\right)  -V(\nu_{\cdot}%
)_{s}\left(  T_{s,0}[\mu_{\cdot}]w\right)  \right\vert \nonumber\\
&  =\int_{0}^{t}ds\int_{\mathcal{T}_{D}\times B_{1}}\mu_{0}(dw)\left\vert
V(\mu_{\cdot})_{s}(T_{s,0}[\mu_{\cdot}]w)-V(\nu_{\cdot})_{s}(T_{s,0}%
[\mu_{\cdot}]w)\right\vert \nonumber\\
&  =\int_{0}^{t}ds\int_{\mathcal{T}_{D}\times B_{1}}\mu_{s}(dw)\left\vert
V(\mu_{\cdot})_{s}(w)-V(\nu_{\cdot})_{s}(w)\right\vert \ .\nonumber
\end{align}
But,
\begin{align}
\left\vert V(\mu_{\cdot})_{s}(w)-V(\nu_{\cdot})_{s}(w)\right\vert  &
\leq|A(x,\mu_{s})-A(x,\nu_{s})|\label{orli2}\\
&  \leq\left\vert \frac{\int_{\mathcal{T}_{D}\times B_{1}}U(x-y)u\mu
_{s}(dy,du)-\int_{\mathcal{T}_{D}\times B_{1}}U(x-y)u\nu_{s}(dy,du)}%
{\int_{\mathcal{T}_{D}\times B_{1}}U(x-y)\mu_{s}(dy,du)}\right\vert
\nonumber\\
&  +\left\vert \frac{\int_{\mathcal{T}_{D}\times B_{1}}U(x-y)\mu
_{s}(dy,du)-\int_{\mathcal{T}_{D}\times B_{1}}U(x-y)\nu_{s}(dy,du)}%
{\int_{\mathcal{T}_{D}\times B_{1}}U(x-y)\mu_{s}(dy,du)}\right\vert \nonumber
\end{align}
Since for any measure $\nu\in{\mathcal{M}},\int_{\mathcal{T}_{D}\times B_{1}%
}U(x-y)\nu_{s}(dy,du)\geq\inf_{x\in{\mathcal{T}}_{D}}U(x)=a,$ we have
\begin{align}
&  \left\vert \frac{\int_{\mathcal{T}_{D}\times B_{1}}U(x-y)\mu_{s}%
(dy,du)-\int_{\mathcal{T}_{D}\times B_{1}}U(x-y)\nu_{s}(dy,du)}{\int
_{\mathcal{T}_{D}\times B_{1}}U(x-y)\mu_{s}(dy,du)}\right\vert \\
&  \leq\frac{\sup_{x\in\mathcal{T}_{D}}|\nabla U(x)|+\sup_{x\in\mathcal{T}%
_{D}}U(x)}{a}d_{b{\mathcal{L}}}(\mu_{s},\nu_{s})\nonumber
\end{align}
and
\begin{align}
&  \left\vert \frac{\int_{\mathcal{T}_{D}\times B_{1}}U(x-y)u\mu
_{s}(dy,du)-\int_{\mathcal{T}_{D}\times B_{1}}U(x-y)u\nu_{s}(dy,du)}%
{\int_{\mathcal{T}_{D}\times B_{1}}U(x-y)\mu_{s}(dy,du)}\right\vert \\
&  \leq\sum_{i=1}^{d}\left\vert \frac{\int_{\mathcal{T}_{D}\times B_{1}%
}U(x-y)u^{i}\mu_{s}(dy,du)-\int_{\mathcal{T}_{D}\times B_{1}}U(x-y)u^{i}%
\nu_{s}(dy,du)}{\int_{\mathcal{T}_{D}\times B_{1}}U(x-y)\mu_{s}(dy,du)}%
\right\vert \nonumber\\
&  \leq d\frac{\sup_{x\in\mathcal{T}_{D}}|\nabla U(x)|+\sup_{x\in
\mathcal{T}_{D}}U(x)}{a}d_{b{\mathcal{L}}}(\mu_{s},\nu_{s})\ .\nonumber
\end{align}
Therefore,
\begin{equation}
\left\vert V(\mu_{\cdot})_{s}(w)-V(\nu_{\cdot})_{s}(w)\right\vert \leq
2d\frac{\sup_{x\in\mathcal{T}_{D}}|\nabla U(x)|+\sup_{x\in\mathcal{T}_{D}%
}U(x)}{a}d_{b{\mathcal{L}}}(\mu_{s},\nu_{s})=\frac{c_{0}}{a}d_{b{\mathcal{L}}%
}(\mu_{s},\nu_{s})\ , \label{orli20}%
\end{equation}
where we have set $c_{0}:=2d(\sup_{x\in\mathcal{T}_{D}}|\nabla U(x)|+\sup
_{x\in\mathcal{T}_{D}}U(x)).$
It is essential that $a>0.$ This is the case for interactions considered in
the Lemmata \ref{le2} and \ref{le3} once the system is confined on the torus
${\mathcal{T}}_{D}$\footnote{In the case where $U$ is with compact support and
$M$ is replaced by $M_{\epsilon}$ we have that
\[
\inf_{x\in X}(U\star\nu)(x)+\epsilon\geq  
\epsilon.
\]
In this case $X$ can be either ${\mathbb{R}}^{d}$ or ${\mathcal{T}}_{D}.$}.
Thus, by (\ref{mea9}), (\ref{mea10}), (\ref{mea11}) and (\ref{orli2}) we have
that
\begin{equation}
\lambda(t)\leq L\int_{0}^{t}\lambda(s)ds+\frac{c_{0}}{a}\int_{0}%
^{t}d_{b{\mathcal{L}}}(\mu_{s},\nu_{s})ds\ . \label{feb2}%
\end{equation}
Hence, since by (\ref{mea9}) $\lambda\left(  0\right)  =0$ we obtain
\begin{equation}
\lambda(t)\leq\frac{c_{0}}{a}\int_{0}^{t}e^{L(t-s)}d_{b{\mathcal{L}}}(\mu
_{s},\nu_{s})ds\ . \label{feb3}%
\end{equation}
Taking in account (\ref{mea2}), (\ref{mea3}), (\ref{mea8}) and (\ref{feb3}) we
get
\begin{equation}
d_{b{\mathcal{L}}}(\nu_{t},\mu_{t})\leq e^{Lt}d_{b{\mathcal{L}}}(\mu_{0}%
,\nu_{0})+\frac{c_{0}}{a}\int_{0}^{t}e^{L(t-s)}d_{b{\mathcal{L}}}(\mu_{s}%
,\nu_{s})ds\ . \label{feb4}%
\end{equation}
Applying the Gronwall's lemma we get bound (\ref{ine1}).

\vskip1.cm

\textit{Step 2} To prove the existence of a solution for the fixed point
equation
\begin{equation}
\mu_{t}=\mu_{0}\circ T_{0,t}[\mu_{\cdot}]\ , \label{mea1}%
\end{equation}
we use the Banach fixed point theorem. Let $\mu$ be the initial condition. To
every curve $\left[0,T\right]\ni t\mapsto\mu_{t}\in\mathcal{M}, \mu_{0}=\mu$, we relate the solution curve
\begin{equation}
\left[0,T\right]\ni t\longmapsto\mu\circ T_{0,t}[\mu_{\cdot}] \in \mathcal{M}
\end{equation}
Let us denote this map ${\mathcal{F}}:C_{{\mathcal{M}}}\rightarrow
C_{{\mathcal{M}}},$ where $C_{{\mathcal{M}}}$ is the space of weakly
continuous function $[0,T]\rightarrow{\mathcal{M}}$ with $\mu_{0}=\mu_{\cdot}$
We equip $C_{{\mathcal{M}}}$ with the metric
\begin{equation}
d_{\alpha}(\mu(\cdot),\nu(\cdot))=\sup_{t\in\lbrack0,T]}\left[  e^{-\alpha
t}d_{b{\mathcal{L}}}(\nu_{t},\mu_{t})\right]  \ ,
\end{equation}
for some $\alpha>0$ which will be suitably chosen. Since $({\mathcal{M}%
},d_{b{\mathcal{L}}})$ is a complete metric space, so is $(C_{{\mathcal{M}}%
},d_{\alpha}).$ Now from Step 1 we have%
\begin{equation}
d_{b{\mathcal{L}}}(\nu_{t},\mu_{t})=d_{b{\mathcal{L}}}({\mathcal{F}}(\mu
(\cdot))(t),{\mathcal{F}}(\nu(\cdot))(t))\leq\frac{c_{0}}{a}\int_{0}%
^{t}e^{L(t-s)}d_{b{\mathcal{L}}}(\mu_{s},\nu_{s})ds \label{feb6}%
\end{equation}
and therefore
\begin{equation}
d_{\alpha}({\mathcal{F}}(\mu(\cdot))(t),{\mathcal{F}}(\nu(\cdot))(t))\leq
\frac{c_{0}}{a(\alpha-L)}d_{\alpha}(\mu(\cdot),\nu(\cdot)) \label{feb7}%
\end{equation}
for $\alpha>L.$ By a suitable choice of $\alpha$ this proves that
${\mathcal{F}}$ is a contraction. \qed

\end{document}